\documentclass[11pt,format=acmsmall, review=false]{acmart}
\usepackage{acm-ec-25}
\usepackage{booktabs} % For formal tables
\usepackage[ruled]{algorithm2e} % For algorithms
\usepackage{geometry}
\SetAlFnt{\small}
\SetAlCapFnt{\small}
\SetAlCapNameFnt{\small}
\SetAlCapHSkip{0pt}
\IncMargin{-\parindent}

\setcitestyle{authoryear}

\title{Tractable Graph Structures in EFX Orientation}

\author{Václav Blažej}
\author{Sushmita Gupta}
\author{M.S. Ramanujan}
\author{Peter Strulo}

\begin{abstract}
Since its introduction, envy-freeness up to any good (EFX) has become a fundamental solution concept in fair division of indivisible goods. Its existence, however, remains elusive---even for four agents with additive utility functions, it is unknown whether an EFX allocation always exists. Unsurprisingly, researchers have explored restricted settings to delineate tractable and intractable cases.
Christadolou, Fiat et al.~[EC'23] introduced the notion of EFX-orientation, where the agents form the vertices of a graph and the items correspond to edges, and an agent values only the items that are incident to it. The goal is to allocate items to one of the adjacent agents while satisfying the EFX condition.
This graph based setting has received considerable attention and has led to a growing body of work.

Building on the work of Zeng and Mehta'24, which established a sharp complexity threshold based on the structure of the underlying graph---polynomial-time solvability for bipartite graphs and NP-hardness for graphs with chromatic number at least three---we further explore the algorithmic landscape of EFX-orientation by exploiting the graph structure using parameterized graph algorithms.

Specifically, we show that bipartiteness is a surprisingly stringent condition for tractability: EFX orientation is NP-complete even when the valuations are symmetric, binary and the graph is at most two edge-removals away from being bipartite. Moreover, introducing a single non-binary value makes the problem NP-hard even when the graph is only one edge removal away from being bipartite. We further perform a parameterized analysis to examine structures of the underlying graph that enable tractability. In particular, we show that the problem is solvable in linear time on graphs whose treewidth is bounded by a constant. Furthermore, we also show that the complexity of an instance is closely tied to the sizes of acyclic connected components on its one-valued edges. 
\end{abstract}

\usepackage{hyperref}
\hypersetup{
    colorlinks=true,
    linkcolor=blue,
    filecolor=magenta,      
    urlcolor=cyan,
    pdfpagemode=FullScreen,
    }

\usepackage{amsmath}
\usepackage{cleveref}
\usepackage{amsthm}
\usepackage{thmtools}
\usepackage{thm-restate}
\declaretheorem{theorem}

\usepackage{mathtools}
\usepackage{graphicx}

\usepackage{xspace}
\usepackage{todonotes}
\usepackage{enumitem}
\theoremstyle{acmplain}
\newtheorem{lemma}[theorem]{Lemma}
\newtheorem{proposition}[theorem]{Proposition}

\newtheorem{observation}[theorem]{Observation}
\Crefname{observation}{Observation}{Observations}

\Crefname{reduction}{Reduction rule}{Reduction rules}

\newtheorem{claim}[theorem]{\textit{Claim}}
\Crefname{claim}{Claim}{Claims}

\Crefname{clm}{Claim}{Claims}
\Crefname{clm2}{Claim}{Claims}

\theoremstyle{acmdefinition}
\newtheorem{definition}[theorem]{Definition}

\usepackage{caption}
\usepackage{subcaption}

\newcommand{\app}{$\!{\rm(\clubsuit)}$\xspace}

\newcommand{\EFX}{\textsc{EFX}\xspace}

\newcommand{\EFXO}{\textsc{EFX-orientation}\xspace}
\newcommand{\EFXp}{\textsc{EFX}$^{+}$\xspace}

\newcommand{\PMSAT}{\textsc{Planar Monotone 3-SAT}\xspace}
\newcommand{\OG}{\overrightarrow{G}}
\renewcommand{\OE}{\overrightarrow{E}}

\newcommand{\sm}{\setminus\!}

\newcommand{\smallcore}{\mathsf{small\mbox{-}core}}
\newcommand{\taucore}{\tau\mathsf{\mbox{-}core}}
\newcommand{\bigcore}{\mathsf{big\mbox{-}core}}

  \newcommand{\D}{\mathcal D}
    \newcommand{\tpp}{\mathrm{top}}
    \newcommand{\btt}{\mathrm{bot}}
    \newcommand{\no}{\mathrm{no}}

\usepackage{tabularx}
\usepackage{tcolorbox}
\newcommand{\defprob}[3]{
    \begin{tcolorbox}[left=6pt,right=6pt]
        \begin{minipage}{0.99\textwidth}
            \begin{tabular*}{\textwidth}{@{\extracolsep{\fill}}ll} #1   \\ \end{tabular*}
            {\bf{Input:}} #2  \\
            {\bf{Question:}} #3
        \end{minipage}
    \end{tcolorbox}
}

\begin{document}

% Title page for title and abstract only.
\begin{titlepage}

\maketitle
\vspace{1cm}
\setcounter{tocdepth}{2} % adjust to 1 if desired

\end{titlepage}

%%%%%%%%%%%%%%%%%%%%%%%%%%%%%%%%%%%%%%%%%%%%%%%%%%%%%%%%%%%%%%%%%%%%%%%%%%%%%%%%

\section{Introduction}

The fair division of resources is a classical problem with deep roots in economics, game theory, and computer science. While fair division of divisible goods—such as land, time, or money—allows for precise proportional splits, the problem becomes significantly more challenging when dealing with indivisible goods, like houses, tasks, or family heirlooms in an inheritance. These goods cannot be split without losing their value, making it difficult to guarantee fairness.

A key challenge in fair division of indivisible goods is that traditional fairness concepts, like proportionality and envy-freeness, may not be attainable even in the simplest scenarios, such as one more agent than items. To address these challenges, researchers have introduced relaxed fairness criteria such as {\it envy-free up to one good} (EF1) \cite{Budish11}, where the goal is to allocation of $m$ indivisible goods among $n$ agents $A_1, \ldots, A_n$ in a manner that for every pair of agents $i$ and $j$, there exists an item $x \in A_i$ such that agent
$j$ does not value $A_i\sm \{x\}$ higher than $A_j$, its own assignment. 
Such allocations are known to always exist, and can be found by the {\it round-robin} protocol~\cite{Caragiannis-EFX/journal} or {\it envy-cycle elimination} protocol~\cite{Lipton04}. In certain scenarios, an EF1 solution can manifest in an absurd unrealistic manner, as demonstrated by the following example: two siblings fighting over their inheritance are not likely to be satisfied if the older gets the family home and grandfather's watch and the younger sibling gets the car, even though the latter is said to be EF1-happy because the removal of the home would make the younger like their bequest better than of the sibling's. Strengthening the condition Caragiannis et al.~\shortcite{Caragiannis-EFX/journal} introduced {\it envy-freeness up to any good} (EFX), where the goal is to create an allocation $A_1, \ldots, A_n$ where for every pair of agents $i$ and $j$, we have that for {\it any} item $x \in A_i$, agent
$j$ values $A_j$ at least as much as $A_i\sm \{x\}$.\footnote{The precise definition in \cite{Caragiannis-EFX/journal} was that of {\it envy-freeness up to any positively-valued good} ~(\EFXp): For any item 
 $x\in A_i$ that agent $j$ values as strictly positive, we have that $j$ values $A_j$ at least as much as $A_i\sm \{x\}$.} In the short time since its introduction~\cite{Caragiannis-EFX/journal}, there has not been a shortage of work directed at solving the EFX problem. In particular, \cite{PR20} show that EFX allocations exist for two agents for general valuations or when they share the same ordinal ranking of goods. A significant breakthrough was attained by Chaudhury et al.~\shortcite{Chaudhury24} which shows that EFX always exists for three
agents when the valuations are additive, and it was later extended to more
general valuations but at the cost of leaving some items unassigned (i.e {\it charity}) Berger et al. \shortcite{BergerCohenFeldmanFiat22} and Akrami et al.~\shortcite{AkramiAlonChaudhury/EC23}.

Due to \cite{Mahara24/journal} it is also known that EFX allocations exist for arbitrary allocations if the number of goods is at most 3
more than the number of agents, that is $m \leq n+3$. Beyond this, the problem has proven to be highly elusive to pin-down even for four agents and additive valuations.  Procaccia in his oft-cited article ``Fair Division’s Most Enigmatic Question"~\cite{Procaccia/article/ACM} has called this a ``major open question" in fair division. 
In light of this problem's difficulty, many papers have explored the problem under various restrictions and one of these is the focus of our work. 

%\subsection*{Our setting}
\medskip
\noindent{\bf Our setting. }In this article, we investigate the problem of allocating $m$ indivisible goods among $n$ agents in the graphical setting as introduced by Christodoulou et al.~\shortcite{christodoulou_fair_2023}. The agents are represented as vertices of a graph and the items are represented by the edges. The agents view the items (i.e edges) they are not incident to as {\it irrelevant}, and thus value them as 0. We study the EFX allocation question for a class of valuation functions associated with a graph. Christodoulou et al.~\shortcite{christodoulou_fair_2023} considered both possibilities: an allocation that allows irrelevant items to be allocated as well as that which does not (calling it an \EFX~{\it orientation}). They showed that while an \EFX allocation always exists and can be computed in polynomial time for general valuation functions, computing an \EFX orientation (formally the problem is called {\EFXO}) is NP-complete even when valuations are symmetric, that is agents that share an edge value it equally.

Since \cite{christodoulou_fair_2023} this setting has been subsequently explored by several works that have extended the \EFX allocation question to related settings. 
Most related to our work are that of Zeng et al.~\shortcite{MehtaZeng2024structureefxorientationsgraphs}, Afshinmehr et al.~\shortcite{Afshinmehr/EFX-A-O-bipartite}, and Bhaskar et al.~\shortcite{Bhaskar24/EFX-A-multigraph}. Zeng et al.~\shortcite{MehtaZeng2024structureefxorientationsgraphs} study the problem of {\it strong \EFX orientation}, where the goal is to decide if a given graph will have an \EFX orientation irrespective of the valuation function. They establish a non-obvious connection with the chromatic number of the graph, and in the process define a complete characterization of strong EFX orientation when restricted to binary valuations. Specifically, they demonstrate the centrality of bipartiteness in the tractability of the EFX problem. Other papers, such as \cite{Bhaskar24/EFX-A-multigraph}, \cite{Afshinmehr/EFX-A-O-bipartite} have since focused on bipartite graphs or multi-graphs to find positive results.
Afshinmehr et al.~\shortcite{Afshinmehr/EFX-A-O-bipartite} study \EFX allocation and orientation questions when restricted to bipartite multi-graphs with additive valuations and show that the former is always guaranteed to exist but the latter may not even for simple instances of multi-trees. They show that that orientation question is NP-hard even when the number of agents is a constant and study parameters such as the maximum number of edges between two vertices and the diameter of the graph. 
Bhaskar et al.~\shortcite{Bhaskar24/EFX-A-multigraph} study EFX allocation in bipartite multi-graphs, multi-tree, and multi-graphs with bounded girth for cancellable monotone valuations.

Note that the aforementioned graphical setting should not be confused with the one studied by Payan et al. \shortcite{PayanSV23} where like us, the vertices are agents, however the graph is used to relax fairness notions such that a fairness condition (such as EF, EF1, EFX) only need to be satisfied between pairs of vertices that share an edge. Some additional related works can be found in the \Cref{additional-rw}. 

\medskip
\noindent{\bf Our contributions.} Our first result starts from the work of \cite{MehtaZeng2024structureefxorientationsgraphs} and probes deeper within the structural aspects of the underlying graph in delineating the tractability landscape of the EFX orientation problem.
In bipartite graphs, \EFXO can be solved even with general additive valuations \cite[Lemma 4.1]{MehtaZeng2024structureefxorientationsgraphs}. However, the complexity of the problem remains unclear, even under {\em binary} additive valuations, when the graph is ``close'' to bipartite—i.e., it can be made bipartite by removing only a few edges.  The minimum number of edge removals required to make a graph bipartite is known as the \emph{min-uncut} number, a term derived from the classic \textsc{Max-Cut} problem. Specifically, removing the edges that do not cross a max-cut—the so-called ``uncut'' edges—yields a minimum edge set whose deletion results in a bipartite graph. 

For readers familiar with parameterized complexity, this value is also referred to as the \emph{edge odd cycle transversal (edge-OCT)} number of the graph since making the graph bipartite by deleting minimum number of edges is equivalent to finding a minimum set of edges intersecting every odd cycle in the graph. This approach of analyzing a problem's complexity based on the input's ``distance from 
triviality'' is a well-established methodology in the area of parameterized complexity. It is specifically designed to delineate the boundaries of tractability for various problems that are NP-hard in general~\cite{AgrawalR22,GuoHN04}. The key idea is to first establish tractability for a class of "trivial" instances and then examine how complexity evolves as we gradually move away from this class—for instance, graphs that are ``almost'' bipartite.

\begin{restatable}{theorem}{tractableEOCTone}\label{cor:tractableEOCTone}
    Every \EFXO instance with binary, symmetric and additive valuations where the input graph has min-uncut number 1, has a solution that can be computed in linear time.
    \end{restatable}

We then prove a complementary lower bound.

\begin{restatable}{theorem}{edgeocttwo}\label{thm:edgeoct2}
    \EFXO is NP-hard in both the following cases.
    \begin{itemize}\item Instances with binary, symmetric and additive valuations where the input graph has min-uncut number equal to 2. 
    % where both uncut edges have the same valuation
\item   Instances with symmetric and additive valuations on graphs with min-uncut number equal to 1 when all edges except one have binary valuations.
\end{itemize}
\end{restatable}

Our next set of results proves another dichotomy, this time focused in a more refined way, on the structure defined by the edges valued as 1. Crucially, we prove that the presence of a $P_5$, an induced path on 5 vertices, in $G_1$ (the subgraph of the input graph $G$ that only contains the edges valued 1) plays a fundamental role in the intractability of \EFXO. We refer the reader to \Cref{sec:prelims} for a formal definition of all terms involved. 

\begin{restatable}{theorem}{threesatnpc}\label{thm:3satnpc}
    \EFXO with binary, symmetric and additive valuations is NP-complete even if each connected component of $G_1$ is $P_2$, $P_3$, or $P_5$, graph $G_0$ has maximum degree 1, and the graph $G$ is planar.
\end{restatable}

This is a surprising connection between the existence of small induced paths in the input and the algorithmic complexity of the problem. To demonstrate that these paths (especially the $P_5$s) are not simply an artifact of the methodologies behind our hardness proof, we give an efficient algorithm for all instances where this type of a substructure is absent.

\begin{restatable}{theorem}{tractablepfive} \label{thm:tractable_on_pfive_free}
    \EFXO with binary, symmetric and additive valuations can be solved in linear time when every connected component of $G_1$ is either $P_5$-free or contains a cycle. 
\end{restatable}
We then conduct a parameterized analysis of the problem and identify parameters of the input graph that when bounded, determine the tractability of the instance.  
We note that having seen (\Cref{thm:3satnpc}) that $G_1$ containing arbitrarily many copies of $P_5$ leads to NP-hardness while $G_1$ being $P_5$-free leads to efficient solvability (\Cref{thm:tractable_on_pfive_free}), we arrive at a natural question: What precisely constitutes a ``complex'' component (like a $P_5$) and how precisely does their {\em number} influence tractability of \EFXO? 

Towards this, we introduce a novel notion -- that of a ``core'' of an {\EFXO} instance--that is crucial in deciding whether an instance has a solution or not, \Cref{lem:core_procedure}. This result is used to design the linear-time algorithm claimed in \Cref{thm:tractable_on_pfive_free}. Moreover, we note that both the number and the size of the cores, defined formally in \Cref{obs:core}, carve out the separation between tractable (\Cref{lem:coreLemma}) and intractable instances (\Cref{thm:wone_complex_components}) of \EFXO. We refer the reader to \Cref{sec:parameterizations} for precise statements and proofs.

Adding to this we also show that the problem is tractable on graphs of small treewidth.

\begin{restatable}{theorem}{fulltreewidth}\label{thm:fulltreewidth}
    \EFXO with binary, symmetric and additive valuations can be solved in linear time on instances $(G,w)$, where the graph $G$ has constant treewidth. 
\end{restatable}

Towards this result, we combine our insights into the problem along with a classic model-checking result in the literature on graph algorithms and parameterized complexity. Such model-checking based algorithmic arguments are, to the best of our knowledge, novel to this line of research and provide a direct and powerful way to obtain algorithms on structurally restricted instances. 

The work of \cite{Deligkas/EFX-EF1-orientations} indicates that the tractability result in \Cref{thm:fulltreewidth} is essentially the best possible for graphs of constant treewidth. More precisely, the requirement of our valuations being binary, symmetric and additive is unavoidable as they show that just dropping the first requirement alone (i.e., binary valuations) makes the problem NP-hard on instances with constant treewidth (in fact, they prove this for a stronger parameter -- vertex cover). 

We conclude the discussion of our contributions by noting that binary valuation represents the most basic--approval--preference model; and our problem setting captures the scenario where every item is either valuable (to two agents) or not to anybody. The additive binary setting is of standalone importance in fair division, and is often the first stepping stone to further analysis as exhibited by several papers~\cite{Garg_Murhekar_Qin_2022,Hadi_Fair_2020,garg2023computing,AZIZ20191,Aziz_2015_Fair_Ordinal} to name a few, that consider valuation functions with one or two non-zero values for various solution concepts such as EF1+ Pareto optimality and Nash social welfare.

\section{Preliminaries}\label{sec:prelims}

\newcommand{\EFEX}[0]{EFX allocation}

Let $[k] = \{1,\dots,k\}$ where $k$ is an integer.
We denote the \emph{agents} as $[n]$ and \emph{items} as $[m]$.
Though both agents and items are referred to by integers, it will be clear which one is meant from the context.
% We reserve $n$ to be always associated with agents and $m$ with items.
A~\emph{bundle} $b$ is a subset of items $b \subseteq [m]$.
For each agent $i \in [n]$ we have a \emph{valuation function} $w_i \colon 2^{[m]} \to \mathbb R_0^+$ which assignes a value to each bundle.
A valuation function is called \emph{binary} if it maps bundles to either 0 or 1, and is called \emph{additive} if $w_i(S)=\sum_{s \in S}w_i(\{s\})$ for every $S \subseteq [m]$.
A~\emph{partial allocation} is a tuple $X = (X_1,\dots,X_n)$ where $X_i \subseteq [m]$ and for every pair of distinct agents $i,j\in [n]$ we have $X_i \cap X_j = \emptyset$.
An allocation $X$ is called \emph{complete} if $\bigcup_{i \in [n]} X_i = [m]$. Unless we specify that an allocation is partial, it will always refer to a complete allocation.   In an allocation $X=(X_1,\dots,X_n)$ an agent $i$ \emph{envies} $j$ if $w_i(X_j) > w_i(X_i)$.
    Agent $i$ \emph{strongly envies} $j$ if there exists $g \in X_j$ such that $w_i(X_j \setminus \{g\}) > w_i(X_i)$.
    An allocation is called EFX if no agent strongly envies any other agent. In the computational problem {\EFEX}, the input instance is  $\mathcal I=([n],[m],\{w_i\}_{i \in [n]})$ and the goal is to produce an EFX-allocation (whenever one can be shown to exist).

Now, we define fair division instances that are representable by a simple undirected graph.

%\il{Over here do we wnat to start by saying "a graph instance" or "an \EFXO instance" ?}
%\begin{definition}[Graph instance]
% \medskip
\noindent{\it Graph instance.} An {\EFEX} instance $\mathcal I=([n],[m]$,$\{w_i\}_{i \in [n]})$ is represented by a graph $G=(V,E)$ if $V=[n]$, $E=[m]$, such that for each vertex $i \in [n]$ and edge $g \in [m]$, if $g$ is incident to $i$, then $w_i(g) \ge 0$, otherwise $w_i(g)=0$. Such {\EFEX} instances are called {\em graph instances}. Moreover,  since each vertex of  $G$  corresponds uniquely to an agent, and each edge corresponds uniquely to an item, we sometimes use the terms agent (item) and vertex (respectively, edge)  interchangeably. 
    
%\end{definition}
%\begin{definition}[Symmetric graph instance]

%\medskip
\noindent{\it Symmetric graph instance.}
Consider an {\EFEX} instance $\mathcal I=([n]$,$[m]$, $\{w_i\}_{i \in [n]})$ represented by a graph $G$. For agents $i$ and $j$ that are adjacent in $G$, we refer to the item corresponding to the edge between them by $ij$. The graph instance is called {\em symmetric} if for every $ij \in E$ we have $w_i(ij) = w_j(ij)$. 
%\end{definition}

In this paper, all  graph instances we work with are symmetric unless otherwise specified. It is clear that graph instances represent {\EFEX} instances where vertices correspond to agents and each agent is incident to the edges corresponding to the set of items they are willing to take into their bundle, and so, each item is only considered by two agents. An allocation that follows this rule is called an {\em orientation} as each item $ij$ can be oriented towards the agent to which it is allocated. The formal definition follows.

%\begin{definition}[Graph orientation]

%\medskip
\noindent{\it Graph orientation.} Consider an {\EFEX} instance $\mathcal I=([n],[m]$,$\{w_i\}_{i \in [n]})$ represented by a graph $G=(V,E)$.  Consider an allocation of the items to agents, where each allocated item $ij\in E$ is given to either $i$ or $j$. This allocation can be represented by the directed graph $\OG=(V,\OE)$, where the edge $ij\in E$ is directed towards $j$ (denoted $(i,j) \in \OE$) when the item is assigned to $j$, and directed towards $i$ when the item is assigned to $i$. An orientation obtained from an EFX-allocation is called an {\em EFX-orientation}. 

%\end{definition}

%\vpsace{-1cm}

From now on we focus only on allocations in the context of graph orientations, i.e., allocations where every item $ij \in E$ is allocated to either $i$ or $j$.
With this assumption, an allocation can be derived from an orientation in the natural way.
For a symmetric graph instance, we may denote the valuation function as $w \colon E \to \mathbb R_0^+$ because both agents incident to an item value the item the same, and all other agents value it 0.
That is, from a graph orientation we can retrieve agent bundles for agent $i \in [n]$ by simply collecting all edges that are oriented towards $i$.
The central computational problem of interest to us in this work is the following.  

% \begin{definition}[Envy and Strong Envy]
%     In an allocation $X=(X_1,\dots,X_n)$ an agent $i$ \emph{envies} $j$ if $w_i(X_j) > w_i(X_i)$.
%     Agent $i$ \emph{strongly envies} $j$ if there exists $g \in X_j$ such that $w_i(X_j \setminus \{g\}) > w_i(X_i)$.
% \end{definition}
% \begin{definition}[Envy-Freeness Up to Any Good (EFX)]
%     An allocation is called EFX if no agent strongly envies any other agent.
% \end{definition}

\defprob{\EFXO}{
    Symmetric graph instance $\mathcal I=([n],[m],w)$ represented by a graph $G=(V,E)$
    % $(G,w)$ of $n$ agents $V$ and $m$ items $E$ and a 
    where  $w \colon E \to \mathbb R_0^+$ is the valuation function.
}{
    Is there a graph orientation $\OG$ of $G$ that corresponds to an EFX-allocation for $\mathcal I$?
}

An EFX-orientation for $\mathcal I$ is called a {\em solution} to the instance $\mathcal I$ of {\EFXO} and instances that have a solution are called {\em yes-instances}. Even though the \EFXO problem appears quite restrictive, we will see that it still is a challenging problem even with simple valuation functions. We focus on \EFXO with an {\em additive binary symmetric} valuation function $w \colon E \to \mathbb \{0, 1\}$. This variant was studied in \cite{christodoulou_fair_2023} where the authors determined that solving \EFXO is an NP-complete problem. In what follows, we show that our search for tractable instances reveals that the problem, \EFXO with additive binary symmetric valuation, is computationally hard even on restricted instances and the search takes us towards instances that have significant restrictions in terms of the underlying graph structure. We use ideas and techniques from parameterized algorithms that have proven to be very powerful in the domain of graph algorithms to exploit the graph structure of a given instance to detect tractable cases and design algorithms wherever possible.

Since $[n]$ and $[m]$ are implied by $G$ in the above definition, we henceforth consider {\EFXO} instances as just a pair  $\mathcal I=(G,w)$. 
In an \EFXO instance we say that an edge is a $c$-{\em edge} if both agents at its endpoints value it at $c$; we use this terminology to refer to 1-edges (or 1-items) and 0-edges (or 0-items) as we focus on symmetric instances with binary valuation functions.
For an instance of \EFXO $(G, w)$, we refer to the subgraph of $G$ induced by the 1-edges as $G_1 = (V,\{uv \mid w(uv)=1, uv \in E\})$ and to the subgraph of $G$ induced by the 0-edges as $G_0 = (V,\{uv \mid w(uv)=0, uv \in E\})$.
Similarly we define $\OG_0$ and $\OG_1$ as directed graphs obtained from an orientation $\OG$ of $G$.
For an agent $u$ and for each $i\in \{0,1\}$, we use the term {\em $i$-neighborhood} to refer to the neighborhood of $u$ within $G_i$, that is, those vertices of $G$ that share an $i$-edge with $u$. The $i$-neighborhood of $u$ is also denoted by $N_i(u)$. Moreover, the {\em $i$-degree} of $u$ refers to $|N_i(u)|$, equivalently, to the number of $i$-edges incident to $u$. 
By {\em induced subgraph} of a graph $G$, we refer to the subgraph induced by a subset $X$ of the vertex set of $G$ and it is denoted by $G[X]$. 
By $P_i$, we denote the path on $i$ vertices. When we say that a graph $G$ contains a $P_i$, we mean that there is an induced subgraph of $G$ that is isomorphic to $P_i$. If a graph does not contain a $P_i$, we say that it is $P_i$-free.  For a graph $G=(V,E)$ and vertex subset $X\subseteq V$, we say that $X$ is an {\em independent set} in $G$ if no edge in $E$ has both endpoints in $X$; equivalently, the subgraph of $G$ induced by $X$ is edgeless.

\section{Basic Properties of EFX-orientations}
In this section, we establish several facts that are used in our work. 

\subsection{Simplifying a given instance}

\begin{restatable}{observation}{FirstObservation}\label{obs:strong_envy_conditions}\app
    Let ${\mathcal I}=(G,w)$ be an instance of {\EFXO}. 
    In an orientation of $G$, an agent $u$ strongly envies an agent $v$ if and only if:
    \begin{enumerate}%[label=(\arabic*)]
        \item\label{envy:one_edge} Edge $(u,v) \in E(\OG_1)$; and 
        \item\label{envy:value_one} agent $u$ gets no items they value as 1; and
        \item\label{envy:other_item} agent $v$ gets another item besides $uv$.
    \end{enumerate}
\end{restatable}

We next describe a series of preprocessing steps designed to simplify \EFXO instances while ensuring that the solution to the original instance remains unchanged. Specifically, these steps guarantee that the original instance admits an EFX-orientation if and only if the preprocessed instance does. Furthermore, any EFX-orientation obtained for the preprocessed instance can be efficiently transformed into a corresponding EFX-orientation for the original instance.

\begin{lemma}\app \label{lem:simple_reductions}
Let ${\mathcal I}=(G,w)$ be an instance of {\EFXO}.  Let ${\mathcal I'}=(G',w')$ be the \EFXO instance resulting from $\mathcal I$ by performing either of the following operations. 
    \begin{enumerate}%[label=(\arabic*)]
        \item\label{it:isolated}
            Let $u$ be vertex whose 1-degree is 0.
%             $G_1$.
            Let $G'=G-u$, and $w'$ be defined as $w$ restricted to $G'$.
        \item\label{it:component}
            Suppose $H$ is a connected component of $G_1$ that contains a cycle $C$.
            Let $G'$ be the graph resulting from deleting the vertices of $H$ from $G$ and $w'$ is defined as $w$ restricted to $G'$.
    \end{enumerate}

    Then, ${\mathcal I}$ has an EFX-orientation  if and  only if ${\mathcal I}'$ has an EFX-orientation. Moreover, given an EFX-orientation for ${\mathcal I}'$, one for ${\mathcal I}$ can be computed in linear time. 
\end{lemma}

When neither of the operations described in \Cref{lem:simple_reductions} can be applied on an instance ${\mathcal I}=(G,w)$,  we call ${\mathcal I}$ a {\em preprocessed instance} and for ease of presentation, when $\mathcal I$ is clear from the context, we simply say that $G$ is {\em a preprocessed graph}.
\paragraph{Remark:}
Due to \Cref{lem:simple_reductions}, when designing algorithms for {\EFXO}, we may assume from now on that the input instance is already preprocessed. Moreover, it is straightforward to see that the transformation from an arbitrary instance to a preprocessed instance can be done in linear time since this just requires one to identify the vertices of degree 0 in $G_1$ along with those connected components that are not trees. 
% \end{remark}

\subsection{Nice EFX-orientations}\label{ss:rooting-trees-in-G1}
Due to our assumption regarding preprocessed instances, it follows that in any instance ${\mathcal I}=(G,w)$, the graph $G_1$ is a forest with trees $T_1,\dots,T_t$ and none of these is an isolated vertex (i.e., a vertex with no edges incident to i).
Moreover, since each tree in $G_1$ has fewer edges than vertices we know that in any orientation at least one agent is allocated no 1-items. If we were to think of one such agent from each tree as the ``root'' of the tree, then one approach to finding an EFX-orientation is to choose an appropriate \emph{root} vertex in each $T_i$ and orient all 1-edges of each tree away from their roots.
Naturally, the roots cannot be chosen arbitrarily. However, it can be argued that whether a given rooting of all the trees in $G_1$ is ``correct'', i.e., there is a corresponding EFX-orientation,  is fully determined by the existence (or the lack thereof) of 0-edges between the neighbors of the roots because the only way a strong envy could manifest is from a root towards its neighbors contained in its tree, as per \cref{envy:other_item} of \Cref{obs:strong_envy_conditions}.
These observations lie behind the following characterization of Zeng and Mehta~\shortcite[Lemma 3.1]{MehtaZeng2024structureefxorientationsgraphs}, which is fundamental to our work.

\begin{proposition}\label{prop:structure}{\rm \cite{MehtaZeng2024structureefxorientationsgraphs}}
    A preprocessed graph $G$ has an EFX-orientation if and only if for each $i \in [t]$, the tree $T_i$ of $G_1$ can be rooted at $r_i \in V(T_i)$ in such a way that the union $R$ of the 1-neighborhoods of the roots in $G_1$, that is, $R = \bigcup_{i \in [t]} N_1(r_i)$, is an independent set in $G_0$.
\end{proposition}

We want to consider the solutions to preprocessed instances that are implied by this characterization.
\begin{definition}[Nice EFX-orientation]
    We refer to EFX-orientations of preprocessed instances where every connected component in $G_1$ has exactly one vertex with no 1-edges oriented towards it as {\em nice} EFX-orientations. We call this vertex the \emph{root} of the connected component. 
\end{definition}

Note that it is possible for a preprocessed instance to have EFX-orientations that are not nice.
For instance, consider the example in \Cref{fig:very_nice}. However, we show later (\Cref{cor:alwaysnice}) that if there is an EFX-orientation, then we may as well assume that it is nice. 

\begin{figure}[ht]
    \centering
    \hfill
    \begin{subfigure}{0.48\textwidth}
        \centering
        \includegraphics[page=1,scale=1.1]{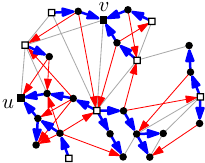}
        \subcaption{EFX-orientation}
        \label{fig:not_nice}
    \end{subfigure}
    \hfill
    \begin{subfigure}{0.48\textwidth}
        \centering
        \includegraphics[page=2,scale=1.1]{nice.pdf}
        \subcaption{Nice EFX-orientation}
        \label{fig:very_nice}
    \end{subfigure}
    \hfill
    \caption{
        EFX-orientations where (\subref{fig:not_nice}) is not nice and (\subref{fig:very_nice}) is nice.
        Thick blue are 1-items, thin red are 0-items.
        0-items are shown in color only when their orientation is pre-determined by having an endpoint in a 1-neighbor of a vertex with no 1-items.
        Note that the set of 0-items with pre-determined direction decreases when an orientation is changed to a nice EFX-orientation.
    }%
    \label{fig:nice}
\end{figure}

\begin{observation}\label{obs:niceEFXObservations}
Consider a preprocessed instance $(\mathcal{I},w)$ of {\EFXO} and let $T_1,\dots, T_t$ be the trees that comprise the connected components of $G_1$. The following hold.

\begin{enumerate}
    \item\label{obs:orient_neighbor_zeros_away}
    In every nice EFX-orientation that roots a tree $T_i$ at vertex $u$, every 0-edge incident to any vertex $v \in N_1(u)$ is oriented away from $v$.
    \item \label{obs:donot_root_when_vertex_has_zeroes}
    In every nice EFX-orientation and for every $i\in [t]$ such that a vertex $u\in V(T_i)$ has a 0-edge oriented toward it, the tree $T_i$ cannot be rooted at any vertex in $N_1(u)$. 
    % \item \label{obs:acummulate_zeroes}
    % If there is an EFX-orientation where an agent $u$ gets a 0-item, then there is an EFX-orientation where  all 0-edges incident to $u$ are oriented towards $u$. \todo{not used}
\end{enumerate}

\end{observation}

\begin{lemma}\label{cor:alwaysnice}\app
    If a preprocessed instance has an EFX-orientation then it has a nice EFX-orientation.
\end{lemma}

Due to \Cref{lem:simple_reductions} and \Cref{cor:alwaysnice}, we may henceforth assume that when we aim to solve a given instance $(G,w)$ of {\EFXO}, (i) the graph $G$ is already preprocessed, and (ii)
% From now on we focus on the above natural question about
it is sufficient for us to find the roots of the trees $T_1,\dots,T_d$ in $G_1$ that a nice EFX-orientation picks.

\section{Dichotomy I: effect of distance from bipartite graphs}

In this section, we show that \EFXO instances with binary valuations, where the underlying graph has a min-uncut number of 1 can be solved in polynomial time.
On the other hand, we show that the problem becomes NP-hard once we step outside of either of these restrictions. Specifically, (i) keeping the binary valuations and increasing the min-uncut number to 2 makes the problem NP-hard, and (ii) allowing even a single non-binary valuation while maintaining a min-uncut number of 1 also leads to NP-hardness.
Note that since we are dealing with graph instances, the valuations are symmetric and so our hardness results hold even in the presence of  symmetric valuations.

\begin{table}[ht]
    \centering
    \caption{Overview of our results based on the min-uncut number and possible restrictions to the valuation function. Here, we say that a valuation function is {\em binary with one exception} if all but one of the edges have a binary valuation on them.}%
    \label{tab:edge_oct_cases}
    \begin{tabular}{ c c c }
        Min-uncut number & valuation function        & complexity \\
        \toprule
        0        & general                   & polynomial time~\cite{MehtaZeng2024structureefxorientationsgraphs} \\
        1        & binary                    & polynomial time,~\Cref{cor:tractableEOCTone} \\
        1        & binary with one exception & NP-hard, \Cref{thm:edgeoct2} \\
        2        & binary & NP-hard,~\Cref{thm:edgeoct2} \\ \bottomrule
    \end{tabular}
\end{table}

\subsection{Orientability and tractability for instances at distance 1 from bipartite graphs}

We show that every instance with min-uncut number 1 has an EFX-orientation and moreover, one can be computed efficiently.

\begin{theorem}
% \begin{restatable}{theorem}{Edge-OCT1-general-condition}
\label{thm:Edge-OCT1-general-condition}
    Consider an \EFXO instance $\mathcal I=(G,w)$, where $V(G)$ can be partitioned into sets $A$ and $B$ such that
    \begin{itemize}
        \item $G[A]$ consists only of 1-edges where each connected component contains at most one cycle, and
        \item $G[B]$ contains only 0-edges.
    \end{itemize}
    Then $\mathcal I$ is a yes-instance and given the sets $A$ and $B$, an EFX-orientation can be found in linear time.
    % $O(m+n)$.
% \end{restatable}
\end{theorem}

% Note that the structure of the edges between $A$ and $B$ can be arbitrary.
\begin{proof}
    We give a constructive proof of the existence of EFX-orientations for such instances, which also implies the claimed algorithm, see \Cref{fig:orient}.
    \begin{description}
        \item[Step 1:]
            We begin orienting each connected component in $G[A]$.
            For each component that contains a cycle, we orient the 1-edges of this cycle around it, making it a directed cycle. The remaining 1-edges of this component are oriented away from the cycle -- such an orientation assigns exactly one edge to every agent. Since there is exactly one cycle in this component, we are done.
            For each component that contains no cycle (is a tree) we choose an arbitrary vertex as the root and orient all 1-edges in $G[A]$ away from the root. This completes the orientation of all edges with both endpoints in $A$.
        \item[Step 2:]
            For each vertex $u$ selected as a root in the previous step, check if there is a vertex $v\in B$ such that $w(uv)=1$ and if so, orient the edge $uv$ towards $u$. 
            % it has at least one 1-edge to a white vertex (which therefore was not considered in the Step 1) we orient one of these 1-edges towards $u$.
        \item[Step 3:]
            Orient every edge in $G[B]$ arbitrarily and every remaining edge (which must have an endpoint each in $A$ and $B$) towards its endpoint in $B$.
    \end{description}
    \begin{figure}[ht]
        \centering
        \includegraphics[scale=1.1]{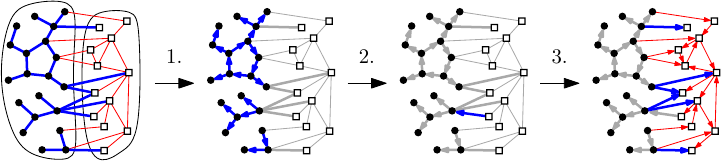}
        \caption{
            \EFXO created during the proof of \Cref{thm:Edge-OCT1-general-condition}.
            Initial graph with two parts $A$ and $B$, and result after application of Step 1, 2, and 3;
            the edges that were oriented by the previous step are in color.
        }%
        \label{fig:orient}
    \end{figure}

    We now argue that the result is an EFX-orientation.
    Since every vertex $w\in A$ gets at most one item (i.e., at most one edge oriented towards it), it cannot be strongly envied, by \Cref{obs:strong_envy_conditions} (3). By \Cref{obs:strong_envy_conditions} (1), no vertex in $B$ can strongly envy another vertex in $B$ since all edges between these vertices are 0-valued. It remains to show that vertices in $A$ do not strongly envy vertices in $B$. 

    Suppose for a contradiction that a vertex $u\in A$ strongly envies $v\in B$. Notice that if  $u$ is not selected as a root in Step 1, then, it gets an item it values as 1 in this step and  by \Cref{obs:strong_envy_conditions} (2), $u$ cannot strongly envy $v$. So, $u$ is a root. If $v$ does not satisfy the check made in Step 2 for $u$, it must be the case that $w(uv)=0$ and by \Cref{obs:strong_envy_conditions} (1), $u$ cannot strongly envy $v$. Otherwise, Step 2 ensures that $u$ gets an item it values as 1, and by \Cref{obs:strong_envy_conditions} (1), it cannot strongly envy $v$. 
 
    It is straightforward to check that given the sets $A$ and $B$ the algorithm runs in linear time. Indeed, it is a linear-time process to compute the unique cycle (if one exists) in each connected component of $G[A]$, after which each edge in the graph is accessed only a constant number of times. 
\end{proof}

Recall that for a graph $G$, a partition of the vertex set into sets $A$ and $B$ is a {\em bipartition} if the graph induced on each of these sets is edgeless. Equivalently, a graph is bipartite if and only if it has a bipartition.

\begin{lemma}\label{lem:corollaryEOCT} \app
    If an \EFXO instance $(G,w)$ has min-uncut number 1, then $G$ has a bipartition $(A,B)$ where the sets satisfy the premise of \Cref{thm:Edge-OCT1-general-condition}. Moreover, given an edge $e$ such that $G-e$ is bipartite, the sets $A$ and $B$ can be computed in linear time.
\end{lemma}

\begin{proposition}{\rm \cite{IwataWY16,RamanujanS17}}\label{linear-time-OCT}
    There is a linear-time algorithm that, given a graph $G$ either correctly answers that the min-uncut number is more than 1 or produces an edge $e$ such that $G-e$ is bipartite.
\end{proposition}

% As a result of \Cref{thm:Edge-OCT1-general-condition}, \Cref{lem:corollaryEOCT}, and \Cref{linear-time-OCT} we obtain the following. 

\tractableEOCTone*

\begin{proof}
    Direct combination of \Cref{thm:Edge-OCT1-general-condition}, \Cref{lem:corollaryEOCT}, and \Cref{linear-time-OCT}. 
\end{proof}

%*************************

\subsection{NP-hardness for instances with min-uncut number more than one}
We now complement our positive result in \Cref{cor:tractableEOCTone} with two intractability results even under strong restrictions. Both intractability results are obtained by slight modifications of a single reduction, so we prove them together in the following theorem. 
% We show that instances where the input graph is bipartite after removal of just 2 edges are also NP-complete, and graphs which are bipartite after removal of 1-edge are solvable in polynomial time.

\edgeocttwo*

\begin{proof}
    We give a reduction from multicolored independent set (MIS), a standard $W[1]$-hard problem \cite{cygan2015parameterized}. 
    In this problem, we are given a graph $G'$, integer $k$, and a $k$-partition $V_1,\dots,V_k$ of $V(G')$ (each $V_i$ is called a color class).
    The goal is to find a set of $k$ mutually non-adjacent vertices $S$ such that $|S \cap V_i|=1$ for every $i \in [k]$.
    This problem is $W[1]$-hard when parameterized by $k$.
    % consisting of a graph $G'$, an integer $k$, and a $k$-partition $V_1,\dots,V_k$ of $V(G')$ where the goal is to find a subset of $k$ mutually non-adjacent vertices $S$.
    We describe the instance of \EFXO, $(G, w)$, that is the result of the reduction, starting with its 1-graph $G_1$.
    For each color $i$ in our MIS instance we produce a path $P_i$ on $4 |V_i|$ vertices.
    Each vertex $u \in V_i$ corresponds to four vertices of the path denoted $u_1$ to $u_4$.
    (We aim to represent the choice of $u \in V_i$ in MIS by rooting $P_i$ in $u_2$ or $u_4$.)
    Additionally, for each vertex $u \in V_i$ we produce an auxiliary path $Q_u$ on five vertices $x_{u,1}$ to $x_{u,5}$.
    Aside from the described vertices, the instance also contains four isolated vertices $a_1$, $a_2$, $b_1$, and $b_2$.
    We now add 0-edges.
    For each vertex $u$ we add 0-edges $u_2 a_1$ and $u_2 a_2$.
    We also add 0-edges $x_{u,3} b_1$, $x_{u,3} b_2$, and $u_3 x_{u,2}$.
    For every $v \in N(u)$ we add the 0-edge $x_{u,4} v_3$.
    \begin{figure}[ht]
        \centering
        \includegraphics[scale=1.1,page=1]{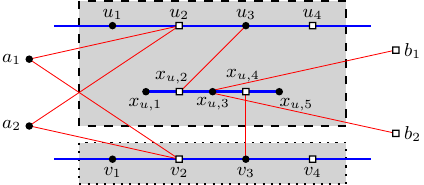}
        \caption{
            An illustration of the part of our construction representing the edge $uv$ in the MIS instance.
            The dashed top gray area represents $u$ on the path for its color, together with its auxiliary path $Q_u$, while the dotted bottom gray area shows 
            part of a (path that represents $v$'s color) that represents $v$.
            Note that in any bipartition of the graph produced in this construction, the vertices depicted by the solid black circles lie in one set while the white squares lie in the other set. 
            % only part of a path that represents $v$.
            % Membership of vertices within the two sides of the bipartition is signified through solid black circles and white squares -- note that edges are only between differently marked vertices.
        }%
        \label{fig:oct_one_connection}
    \end{figure}
    We will later add a gadget using the set $\{a_1,a_2,b_1,b_2\}$ to ensure that all EFX-orientations (if any exist) have a certain property that we call {\em goodness}.
    More precisely speaking,
    we say an EFX-orientation of $G$ is \emph{good} if every vertex that has  0-edges incident to itself and to both $a_1$ and $a_2$ or 0-edges incident to itself and both $b_1$ and $b_2$, is assigned at least one of these 0-edges.
    
    \begin{claim}\label{lem:reductionindependence}\app
        In a good nice EFX-orientation of $G$, for all edges $uv \in E(G)$, if $u_2$ or $u_4$ are chosen as a root then neither $v_2$ nor $v_4$ can be chosen as a root.
    \end{claim}

    \begin{claim}\label{clm:goodiffmis}\app
        There exists a good nice EFX-orientation of $G$ if and only if there exists a multicolored independent set of $G'$.
    \end{claim}

    It remains to describe a gadget, $H$, using the vertices $a_1$, $a_2$, $b_1$, and $b_2$ that ensures that all EFX-orientations of $G \cup H$ are good.
    Note that the construction thus far is bipartite: it consists of paths that are joined either with 2-paths or single edges but only at even distances therefore creating only even cycles (see \Cref{fig:oct_one_connection}). Moreover, note that the construction thus far also results in a connected graph and so, has a unique bipartition up to swapping the sides. 
    Our next step is to design a gadget that ``relaxes'' or ``breaks'' this bipartiteness (or we would be contradicting \Cref{thm:Edge-OCT1-general-condition}).
    We will give three gadgets, each with slightly different structure so that 
    each corresponds to a different relaxation of bipartiteness. 

    The first gadget is simply two 1-edges: $a_1a_2$, and $b_1b_2$.
    Both of these edges increase the min-uncut number so with this gadget the instance will have min-uncut number two.
    Since every edge can be oriented towards only one of its two endpoints, one of $a_1$ or $a_2$ will be assigned the 1-edge $a_1a_2$.
    As $a_1$ and $a_2$ are twins (have identical neighborhoods), suppose without loss of generality that we assigned the 1-edge to $a_1$.
    Then $a_2$ will be envious of $a_1$ if $a_1$ receives any 0-edges.
    Therefore, all 0-edges incident to $a_1$ must be oriented away from $a_1$, and hence, any vertex with a 0-edge to $a_1$ will be assigned a 0-edge.
    The same argument applies to $b_1$ and $b_2$.

    The second gadget has two 1-edges $a_1b_1$, $a_2b_2$ and two 0-edges $a_1a_2$ and $b_1b_2$.
    Again with this gadget the instance will have min-uncut number two but this time the edges not across the cut have value zero.
    Similarly to the first gadget the 1-edges are oriented towards only one of their endpoints so there are two vertices that do not receive a 1-edge and therefore that could potentially be strongly envious.
    This will happen if the two vertices that do receive their incident 1-edges also receive a 0-edge.
    So the vertices that do receive 1-edges cannot be both $a_1$ and $a_2$, or both $b_1$ and $b_2$.
    Furthermore these vertices also must receive no further 0-edges from outside the gadget.
    That is, either every 0-edge incident to $a_1$ or $b_2$ is oriented away from that vertex or similarly for $a_2$ and $b_1$.

    The third gadget is slightly more complicated.
    Again we have two 1-edges, this time $a_1b_1$ and $b_2a_2$.
    Additionally we have the edge $b_1b_2$ valued at $m$ where $0 < m < 1$.
    This is the only edge contributing to the min-uncut number but this edge is not valued 0 or 1 so we do not quite have a binary valuation.
    We consider the possible orientations of the edges between $a_1$, $a_2$, $b_1$, and $b_2$.
    By symmetry we can assume that the $m$ valued edge is directed from $b_1$ to $b_2$.
    Furthermore, if $a_2b_2$ was oriented towards $b_2$, then $b_2$ would get two items while $a_2$ would have only 0-items so $a_2$ will be envious of $b_2$, see (\subref{fig:oct_one_gadgets_d}.I) of \Cref{fig:oct_one_gadgets}.
    So $a_2b_2$ must be oriented towards $a_2$.
    We can now consider the two cases for the orientation of $a_1b_1$.
    \begin{description}[wide=0pt]
        \item[Case 1:] Towards $b_1$, see (\subref{fig:oct_one_gadgets_d}.II) of \Cref{fig:oct_one_gadgets}.
            The vertex $a_1$ may be assigned 0-edges but will have a bundle of value zero so will strongly envy $b_1$ unless $b_1$ gets no 0-items, i.e., all 0-edges incident to $b_1$ are directed away from it.
            Furthermore, $b_2$ gets a bundle with value only $m < 1$, so similarly, $b_2$ will strongly envy $a_2$ unless $a_2$ gets no 0-items.
        \item[Case 2:] Towards $a_1$, see (\subref{fig:oct_one_gadgets_d}.III) of \Cref{fig:oct_one_gadgets}.
            The vertex $b_1$ gets no items so both $a_1$ and $b_2$ must get no 0-items.
    \end{description}
    
    This completes the proof of \Cref{thm:edgeoct2}. 
\end{proof}
\begin{figure}[ht]
    \centering
    \hfill
    \begin{subfigure}{0.18\textwidth}
        \centering
        \includegraphics[page=3,scale=1.1]{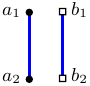}
        \subcaption{}
        % \subcaption{\EFXO instance is created by 1-to-1 mapping; note that labels on left two vertices are swapped to accommodate planarity.}%
        \label{fig:oct_one_gadgets_a}
    \end{subfigure}
    \hfill
    \begin{subfigure}{0.18\textwidth}
        \centering
        \includegraphics[page=4,scale=1.1]{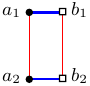}
        \subcaption{}
        \label{fig:oct_one_gadgets_b}
    \end{subfigure}
    \hfill
    \begin{subfigure}{0.18\textwidth}
        \centering
        \includegraphics[page=5,scale=1.1]{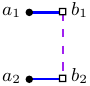}
        \subcaption{}
        \label{fig:oct_one_gadgets_c}
    \end{subfigure}
    \hfill
    \begin{subfigure}{0.41\textwidth}
        \centering
        \includegraphics[page=6,scale=1.1]{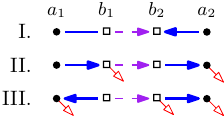}
        \subcaption{}
        \label{fig:oct_one_gadgets_d}
    \end{subfigure}
    \hfill
    \caption{
        (\subref{fig:oct_one_gadgets_a}) The first gadget to ensure min-uncut number 2 in our construction. Note that both of the uncut edges are 1-edges.
        % that forces a good \EFXO.
        (\subref{fig:oct_one_gadgets_b}) The second gadget is similar to the first, but forces two uncut 0-edges in our construction. 
        (\subref{fig:oct_one_gadgets_c}) The third gadget which ensures min-uncut number 1 in our construction. Note that the dashed uncut edge does not have a binary valuation. 
        (\subref{fig:oct_one_gadgets_d}) Figures corresponding to the case analysis for the third gadget. Empty arrows show forced direction for respective 0-edges.
    }%
    \label{fig:oct_one_gadgets}
\end{figure}

%%%%%%%%%%%%%%%%%%%%%%%%%%%%%%%%%%%%%%%%%%%%%%%%%%%%%%%%%%%%%%%%%%%%%%%%%%%%%%%%

\section{Dichotomy II: effect of structure underlying $G_0$ and $G_1$} \label{sec:dichotomypaths}

In this section, we identify a surprising effect that short induced paths
% on 5 vertices (i.e., $P_5$s) 
have on the complexity of {\EFXO}.

Recall that the NP-completeness {of \EFXO} was established by Christodoulou et al.~\shortcite{christodoulou_fair_2023} by a reduction from Boolean circuit satisfiability.
We note that their reduction results in a graph where $G_1$ consists of induced paths on 1, 3, and 5 vertices and does not explicitly restrict the structure of $G_0$.
In this section, {we show that induced paths on 5 vertices, denoted by $P_5$, are in some sense, essential for this NP-hardness, and any (preprocessed) instance where $G_1$ is $P_5$-free can be solved in polynomial time.}
Further, we apply the framework of parameterized complexity to identify those instances with parametrically-restricted structure that can be solved efficiently.

\subsection{Preparing for NP-hardness: reducing 0-degrees}

\begin{lemma}\label{red:zero_degrees}\app
    Let ${\mathcal I}=(G,w)$ be an instance of {\EFXO}.  Let ${\mathcal I'}=(G',w')$ be the \EFXO instance resulting from $\mathcal I$ by performing the following operation.

    Pick a vertex $u \in V(G)$ such that $|N_0(u)| > 1$ and for each of its 0-neighbors $v \in N_0(u)$, delete the $uv$ edge and create an auxiliary arbitrary binary tree $B$ with leaves as $N_0(u)$ and with root in a new vertex $u'$.
    Next add the vertices $V(B) \setminus (N_0(u) \cup \{u\})$ to the graph $G$.
    Now, for each edge $ab \in E(B)$, where $a$ is the vertex closer to $u'$ in $B$, we create a new vertex $w$ and add it to $G$ along with a 1-edge $aw$ and a 0-edge $wb$.
    To finish, we add $uu'$ 0-edge to $G$, see \Cref{fig:high_0_degree} for illustration.
    Then, ${\mathcal I}$ has an \EFX orientation if and only if ${\mathcal I}'$ has an \EFX orientation. Moreover, given an \EFX orientation for ${\mathcal I}'$, one for ${\mathcal I}$ can be computed in linear time. 
 
\end{lemma}

\begin{figure}[ht]
    \centering
    \includegraphics[scale=1.1,page=2]{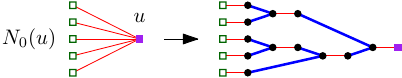}
    \caption{Reduction that replaces high 0-degree vertex with an equivalent binary tree. Recall that the thick blue lines are 1-edges and the thin red lines are 0-edges.}%
    \label{fig:high_0_degree}
\end{figure}

Note that by applying \Cref{red:zero_degrees} on a vertex $u$, the 0-degree of $u$ drops to 1.
By repeatedly applying this operation exhaustively, we obtain a graph where $G_0$ is a matching and to $G_1$ we added only a set of disjoint paths on tree vertices. This subroutine will come in handy in the next section where we aim to obtain ``NP-hard instances'' that additionally possess a significant amount of structure.

\subsection{NP-hardness: presence of small paths in $G_1$}

\threesatnpc*

The proof of this theorem appears at the end of this subsection since we first need to prove an intermediate lemma. Our overall strategy is to give a polynomial-time reduction that is inspired by the reduction of  \cite{christodoulou_fair_2023}. 
We reduce to {\EFXO} from \PMSAT, a variant of 3-SAT.  In a \emph{monotone} CNF formula, each clause contains only positive or only negative literals. A \emph{Monotone rectilinear representation} of such a formula is a planar drawing (if one exists) where variables are drawn as horizontal line segments along the $x$-axis, positive clauses are drawn as horizontal line segments with positive $y$-coordinate, negative clauses are horizontal line segments with negative $y$-coordinate, and whenever a variable is contained in a clause the respective horizontal segments are connected with a vertical segments, see \Cref{fig:rectilinear_3SAT_sat}.
In the \PMSAT problem, the input is a  monotone rectilinear representation of a planar monotone 3-CNF formula and the goal is to decide whether it is satisfiable. This SAT variant is NP-complete, as shown by {de Berg} and Khosravi~\cite[Theorem 1]{deberg_planar_sat_2012}. Towards \Cref{thm:3satnpc}, we prove the following intermediate hardness result. Note that in contrast to \Cref{thm:3satnpc},  the following statement does not make any assertions regarding the 0-degrees of vertices in the reduced instance. However, this will be achieved later by invoking \Cref{red:zero_degrees}. 

\begin{lemma}\label{lem:efxosmallnpc}
    There is a polynomial-time reduction from \PMSAT to \EFXO where the graph $G$ in the produced instance is planar and moreover, each connected component of $G_1$ is either a $P_2$ or a $P_5$.
\end{lemma}
\begin{proof}
    % We first show a reduction from standard 3-SAT to a non-planar instance of \EFXO.
    Let $\phi$ be a planar monotone 3-CNF formula
    % in conjunctive normal form 
    with clauses $\mathcal C = C_1,\dots,C_m$ and variables $\mathcal X = x_1,\dots,x_n$ such that each $C_j$ contains at most three literals $c_{j,1},c_{j,2},c_{j,3}$.
    We build the graph for the instance of \EFXO as follows, see \Cref{fig:hardness_3SAT}.
    \begin{itemize}[wide=0pt]
        \item For each variable $i \in [n]$ create a $P_2$ of weight 1 with vertices named (suggestively) $x_i$ and $\neg x_i$.
        \item For each clause $j \in [m]$ create a $P_5$ of weight 1 with vertices named (first-to-last): \\ $c_{j,1},c_{j,2},g_{j,3},g_{j,2},c_{j,3}$; add 0-edges $c_{j,2}g_{j,2}$ and $g_{j,3}c_{j,3}$.
        \item Define a function $\alpha \colon V \to V$ which maps vertices that represent clause literals to their neighbors within their $P_5$ as follows $\alpha(c_{j,1}) = c_{j,2}$, $\alpha(c_{j,2}) = c_{j,1}$, and $\alpha(c_{j,3}) = g_{j,2}$.
        \item For every pair $i \in [n], j \in [m]$ if for some $q \in [3]$ we have $c_{j,q}=x_i$, i.e., $q$-th literal in $j$-th clause is $x_i$, then connect $\alpha(c_{j,q})$ to vertex $x_i$ with a 0-edge. Similarly if $c_{j,q} = \neg x_i$ then connect $\alpha(c_{j,q})$ to $\neg x_i$ with a 0-edge.
    \end{itemize}
    \begin{figure}[ht]
        \centering
        \includegraphics[scale=1.1]{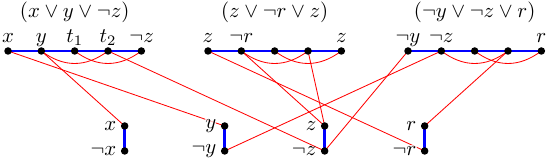}
        \caption{
            Reduction example for a 3-CNF formula $\phi=(x \lor y \lor \neg z) \land (z \lor \neg r \lor z) \land (\neg y \lor \neg z \lor r)$.
            Note that the first clause gadget cannot be rooted in $t_1$ nor $t_2$ due to its 0-edges.
            Rooting the first clause gadget in $x$ forbids rooting the variable gadget in $\neg x$.
        }%
        \label{fig:hardness_3SAT}
    \end{figure}

    Due to \Cref{prop:structure} now it remains to show that the 3-CNF formula $\phi$ is satisfiable if and only if the 1-trees in the resulting instance can be rooted.
    In the forward direction, let us have a satisfying variable assignment.
    Hence, for each clause $C_j$ let $c_{j,q_j}$ be a literal that evaluates to true, and let us root the $j$-th clause gadget $P_5$ at the vertex named $c_{j,q_j}$.
    Next, we root the variable gadget at $x_i$ if $x_i$ is true, otherwise root it at $\neg x_i$.

    We do not root any $C_j$ gadget at either $g_{j,3}$ nor $g_{j,2}$ so 0-edges within the $C_j$ gadget do not have both endpoints adjacent to the root.
    It remains to show correctness for the edges between clause and variable gadgets.
    Towards a contradiction, suppose an edge from $\alpha(c_{j,q})$ to $x_i$ witnesses that the rooting is not a solution, then we must have rooted $C_j$ in $c_{j,q}$ and the variable gadget in $\neg x_i$.
    However such a 0-edge was created because $c_{j,q}=x_i$ and the rooting was based on true valuations so we have that both $x_i$ and $\neg x_i$ would need to be true, a contradiction.

    To show the other direction we recall that each 1 tree needs to be rooted.
    Clause gadgets can be rooted only at $c_{j,1}$, $c_{j,2}$, or $c_{j,3}$ and variable gadgets can be rooted at either of their vertices.
    By the argument above we see that any rooting cannot choose roots that would represent a non-satisfying assignment.
    It is then straight-forward to transform the rooting into a satisfying assignment using vertex labels of the variable gadgets.

    \noindent\textbf{Planarity.}
    Assume each clause has literals initially ordered left-to-right according to the monotone rectilinear representation.
    As the first step, we swap first two literals in each clause so that clause $(x \lor y \lor z)$ is now $(y \lor x \lor z)$.
    Then we place clause and literal gadgets on positions of their horizontal segments from the monotone rectilinear representation.
    Each 0-edge that models literals follows closely the vertical connections of $\mathcal I$, resulting in a planar instance of \EFXO, see \Cref{fig:rectilinear_3SAT_efxo}.
    
    \begin{figure}[ht]
        \centering
        \hfill
        \begin{subfigure}{0.48\textwidth}
            \centering
            \includegraphics[page=1,scale=1.1]{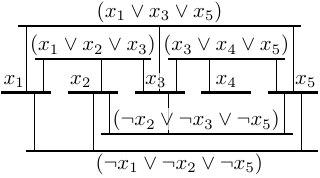}
            \subcaption{\textsc{Planar Monotone 3-SAT} consists of non-crossing orthogonal segments; each clause is monotone.}%
            \label{fig:rectilinear_3SAT_sat}
        \end{subfigure}
        \hfill
        \begin{subfigure}{0.48\textwidth}
            \centering
            \includegraphics[page=2,scale=1.1]{rectilinear_3SAT.pdf}
            \subcaption{The reduced {\EFXO} instance. Note that labels on left two vertices of every clause gadget $P_5$ are swapped to accommodate planarity.}%
            \label{fig:rectilinear_3SAT_efxo}
        \end{subfigure}
        \hfill
        \caption{
            A figure depicting the reduction from \PMSAT to planar 0/1 \EFXO.
            The reduced instance is created essentially by a 1-to-1 mapping from the given instance of {\sc Planar Monotone 3-SAT} so that planarity is maintained.
        }
        \label{fig:rectilinear_3SAT}
    \end{figure}
    
    Note that 0-edges within each clause gadget are easy to draw last and 0-edges connected to a single clause do not cross as $\gamma(C_{j,1})$, $\gamma(C_{j,2})$, and $\gamma(C_{j,3})$ are ordered left-to-right due to the swap of the first two clause literals.
\end{proof}

In order to prove \Cref{thm:3satnpc} one simply replaces vertices with 0-degree more than 1 using \Cref{lem:efxosmallnpc} while preserving planarity, details for this proof can be found in the appendix.

\subsection{Tractability on $P_5$-free instances}
We now complement the hardness in \Cref{thm:3satnpc} by showing that excluding $P_5$'s is sufficient to obtain tractable instances. This implies that the $P_5$s created in the hardness reduction are not just an artefact of our methodology, but are fundamental to the algorithm complexity of our problem.

\begin{observation}\label{obs:dominated_states}
    Consider a preprocessed instance $(\mathcal{I},w)$ of {\EFXO}, let $T_1,\dots, T_t$ be the trees that comprise the connected components of $G_1$.
    For each $i \in [t]$ if we have two vertices $u,v \in T_i$ such that $N_1(u) \subset N_1(v)$ then if there is an EFX-orientation that roots $T_i$ at $v$ then there exists an EFX-orientation with the same edge orientations with one exception -- it roots $T_i$ at $u$ instead.
\end{observation}
\begin{proof}
    Due to \Cref{prop:structure} to solve the preprocessed instance we aim to root all trees $T_1,\dots,T_t$ in $r_1,\dots,r_t$ in such a way that $R=\bigcup_{i \in [t]} N_1(r_i)$ is an independent set in $G_0$, suppose we have such \EFXO.
    Consider $u,r_j \in T_j$ with $N_1(u) \subset N_1(r_j)$, then rooting $T_j$ in $u$ instead of $r_j$ makes $\{u\} \cup (\bigcup_{i \in [t]} N_1(r_i)) \setminus \{r_j\} \subseteq R$, hence, it is still an independent set and the rooting implies an \EFXO.
\end{proof}

In such a case, we say that rooting $T_i$ at $u$ \emph{dominates} rooting $T_i$ at $v$, we call $u$ \emph{dominating} and $v$ \emph{dominated}.
If $N_1(u) = N_1(v)$, then we say that $u$ dominates $v$ if and only if $u$ is before $v$ on the input (arbitrary ordering is sufficient).
Let \emph{number of states} for a tree $T_i$ be defined as the number of its vertices that are not dominated, i.e., number of meaningfully distinct roots.

\begin{definition}\label{obs:core}
    A tree is called a \emph{core} when none of its vertices is adjacent to two leaves.
    A \emph{Core} of a tree is a maximal subtree that is a core.
\end{definition}

By \Cref{obs:dominated_states} two leaves in tree $T_i \in G_1$ are interchangeable within a rooting.
As we see next, one can create an equivalent instance with the extra leaves removed.
We then show (also in the next section) that number of vertices within the \EFXO instance cores is closely tied to its complexity.

\begin{lemma}\label{lem:core_procedure}\app
    There is a linear-time procedure that takes an \EFXO instance $\mathcal I=(G,w)$ as input and produces an \EFXO instance $\mathcal I'=(G',w')$, where each tree in $G'_1$ on at least 4 vertices is a core and moreover, ${\mathcal I}$ has an \EFX~ orientation  if and  only if ${\mathcal I}'$ has an EFX orientation.
\end{lemma}

% The following result highlights the significance of $P_5$s in $G_1$ in contributing to the intractability of \EFXO.

We assume that for a given instance, we have reduced the instance through applications of \Cref{lem:simple_reductions} and \Cref{lem:core_procedure}. So in linear time, we have obtained an equivalent instance in which every component in $G_1$ is a tree with at least two vertices. The latter, also a linear-time procedure, yields an equivalent instance where every tree in $G_1$ with at least four vertices is a core. 

\tractablepfive*

\begin{proof}

    We proceed with a proof that each component of $G_1$ has at most two non-dominated rootings -- we call these \emph{states} in this proof.
    Then we model these two states using a variable and we model each edge of $G_0$ using a clause with 2 literals, ultimately showing a reduction to 2-SAT, which is solvable in linear time.
    The details can be found in the appendix.
\end{proof}

\section{Parameterized Analysis} \label{sec:parameterizations}

In this section, we further explore the boundaries of tractability of {\EFXO} by conducting a parameterized analysis of the complexity of the problem.

\subsection{Parameterization by Treewidth}\label{sec:treewidth}

Treewidth is a structural parameter that provides a way of expressing the resemblance of a graph to a forest. Essentially, it can be seen as a measure of how similar a graph is to a tree; trees have treewidth $1$, while the complete $n$-vertex graph has treewidth $n-1$. A formal definition of treewidth will not be necessary to obtain our results. We will make use of a classic model-checking result of Courcelle \cite{Courcelle90,ArnborgLS91} which says that if a problem is expressible in a specific fragment of logic, then it can be solved in linear time on graphs of constant treewidth. To use this result, we show that {\EFXO} is indeed one such problem.
An brief primer to the notation used in the following theorem can be found in the appendix.

\begin{proposition}[Courcelle's Theorem~\cite{Courcelle90,ArnborgLS91}]\label{fact:MSO} 
Consider a fixed MSO formula $\Phi(x_1,\dots,x_\ell, X_1,\dots, X_q)$ with free individual variables $x_1,\dots,x_\ell$ and free set variables $X_1,\dots,X_q$, and let $w$ be a
    constant. Then there is a linear-time algorithm that, given a labeled
    graph $G$ of treewidth at most $w$, either outputs  $z_1,\ldots, z_\ell\in V(G)\cup E(G)$ and $S_1, \dots, S_q	\subseteq V(G)\cup E(G)$ such that $G \models \Phi(z_1,\ldots, z_\ell, S_1, \dots, S_q)$ or correctly identifies that no such elements $z_1,\ldots, z_\ell$ and sets $S_1, \dots, S_q$ exist.
\end{proposition}

\fulltreewidth*

\begin{proof}
As already argued, we may assume that we have obtained a preprocessed {\EFXO} instance in linear time. For ease of presentation, we use $(G,w)$ to refer to the preprocessed instance as well. 
By \Cref{cor:alwaysnice}, it is then sufficient to compute a nice EFX-orientation, which in turn is equivalent to checking whether there is a set $X$ of vertices -- one from each connected component of $G_1$ such that
$R = \bigcup_{x\in X} N_1(x)$ is an independent set in $G_0$. This is equivalent to checking whether the following statement holds on the input:

\begin{quote}{\em There is a set $X$ of vertices such that every connected component of $G_1$ contains at least one vertex from $X$ and at most one vertex from $X$ and for every pair of vertices $x,y\in X$ (not necessarily distinct) and vertices $z\in N_1(x)$ and $w\in N_1(y)$, there is no 0-edge $zw$. }
\end{quote}

It is easy to see that this statement can be expressed in MSO logic and the rest would follow by invoking \Cref{fact:MSO}. The preprocessed instance $(G,w)$ can be treated as a graph $G$ that has one of two possible labels on each edge (0 or 1) depending on the valuation given by  $w$. This is how we are able to incorporate into our final MSO formula, smaller subformulas that refer specifically to properties of either $G_0$ or $G_1$ alone, for instance the 0-neighborhood of a vertex.  The low level details of the overall formula are tedious but straightforward, so we omit them here. \end{proof}

\subsection{Parameterization by number of long induced paths in $G_1$}\label{ss:cores}

In the previous sections we have seen that $G_1$ containing arbitrarily many copies of $P_5$ leads to NP-hardness, whereas if $G_1$ is $P_5$-free (after preprocessing) then the instance is solvable in polynomial time. This leads to a natural question:

{\it What precisely constitutes a ``complex'' component (like a $P_5$) and how precisely does their number influence tractability of \EFXO?}

First, we process all trees according to \Cref{lem:core_procedure} to create their cores.
Let us partition trees of $G_1$ according to their cores into four groups based on a threshold $\tau$ as follows.
\begin{enumerate}
    \item $\smallcore$: Core of diameter at most 3;
    \item $\taucore$: Core of size at least 4 but at most $\tau$; and 
    \item $\bigcore$: Core of size at least $\tau+1$.
\end{enumerate}

We devise a general algorithm (see \Cref{lem:coreLemma} below) whose running time depends on the number of cores within the various families of cores. Note that the threshold $\tau$ that defines the size of the cores in $\taucore$ is ``user defined'' and is instantiated when this algorithm is run.

\begin{lemma}\label{lem:coreLemma}
    For every $\tau\in {\mathbb N}$, an \EFXO instance with $|\taucore| \le k$ and $|\bigcore| \le b$ is solvable in time $\tau^k \cdot n^{\mathcal O(b)}$.
    % for some computable function $f$.
\end{lemma}
\begin{proof}
    % We apply a simple brute force algorithm.
    Take each tree in $\taucore$ and try all at most $\tau$ root possibilities in $\tau^k \cdot n^{\mathcal O(1)}$ time, and then for each tree in $\bigcore$ in $n^b \cdot n^{\mathcal O(1)}$ time.
    When we have fixed roots in all these trees we can orient all 0-edges that neighbor a root away from that component as per \Cref{obs:niceEFXObservations} (1), and every other 0-edge towards the component as their endpoints in the component do not neighbor the root. %\Cref{obs:orient_neighbor_zeros_away}
    The remaining instance contains only $\smallcore$ which by \Cref{thm:tractable_on_pfive_free} is solvable in $\mathcal O(n+m)$ time.
    In total, this takes $\tau^k \cdot n^{\mathcal O(b)}$ time.
\end{proof}

An obvious question is whether more sophisticated algorithmic techniques could be used to improve the running time of these algorithms. In the next subsection, we show that one cannot do much better than these algorithms.

\subsection{Parameterized hardness: The case of many \sf{big\mbox{-}cores}.}\label{ss:bigcores}

In this part of the paper, by exploring the structural aspects of tree cores, we are able to describe the structure on $\bigcore$, \Cref{lem:whard_options}, which allows us to identify a class of instances that cannot have a better algorithm than the ones presented in the previous subsection. These structures allow us to build hardness gadgets that formalize this insight, \Cref{thm:wone_complex_components}. Specifically, we can rule out (under some complexity-theoretic assumptions) algorithms with running time $f(k)\cdot n^{\O(1)}$, where $k$ denotes the number of $\bigcore$ trees in $G_1$.

We will see later, in \Cref{lem:bijection}, that this structure, namely {\it leafed split orientation}, is closely tied to the notion of \textit{Maximum Induced Matching} (MIM).
To present the reduction we first need several auxiliary notions.
To begin, in \Cref{lem:mim_size_relation}, we establish a relation that ties the size of a tree core to size of its maximum induced matching. Note that proving the statement for tree cores is crucial as many leaves attached to the same vertex can increase size of a tree arbitrarily while not increasing its maximum induced matching. We recall that \emph{eccentricity} of a vertex is its largest distance to all other vertices in the graph. 

\begin{lemma}\label{lem:mim_size_relation}\app
    A tree core with maximum induced matching equal to $m \ge 2$ has size between $2m+1$ and $5m-1$.
\end{lemma}

The next result allows us to bound the size of the maximum induced matching in a tree core in terms of two maximum induced matchings that saturate the leaves of tree; and will be used in \Cref{lem:max_leaf_so} to bound the size of a tree core that contains the structure of interest. 

\begin{lemma}\label{lem:leafedmimisniceandfast}\app
   Let $T$ be a tree core. Suppose $L_1$ and $L_2$ {partition} the set $L$ of leaves in $T$, where for each $i \in [2]$, and for each pair of distinct leaves $u,v \in L_i$ the distance between $u$ and $v$ is at least 4.
   Then, $T$ has maximum induced matchings $R_1$ and $R_2$ (called \emph{leafed} matchings) that saturate $L_1$ and $L_2$, respectively.  Moreover, $|R_1| + |R_2| \ge |M|$,  where $M$ is a maximum induced matching of $T$; and there is an algorithm that finds $R_1$ and $R_2$ in $\mathcal O(|T|)$ time.
\end{lemma}

We are now ready to define the structure that leads us to the intractable instances.

\subsubsection{Split orientation: The structure of interest}
We define a \emph{split orientation} of an undirected graph $G$ to be a set $M$ that consists of ordered pairs of adjacent vertices of $G$.
That is, for each $(u,v) \in M$ we have $uv \in E(G)$; and for each pair of arcs $(u,v),(x,y) \in M$, that are not necessarily vertex disjoint, we have the condition that $uy \not\in E(G)$ and $xv \not\in E(G)$.
Note that these constraints imply that no two arcs in $M$ can share their heads or tails\footnote{For an arc $a=(u,v)$ we call $u$ the head of $a$ and $v$ the tail.}; that is, we cannot have both arcs $(u,v)$ and $(u,y)$ (or $(u,v)$ and $(x,v)$) in $M$.

We will see later that a split orientation of a tree  describes the number of ``sufficiently distinct'' rootings.
Towards that, we will first establish how to determine split orientations via its close relation to maximum induced matchings in a suitable defined product graph. 

We define the {\it direct product of two graphs} $G_1$ and $G_2$, denoted by $G_1\times G_2$, as follows: The vertex set $V(G_1 \times G_2) = V(G_1) \times V(G_2)$ and a pair of vertices $(a_1, b_1) , (a_2, b_2)$ is an edge in $G_1\times G_2$ if $a_1a_2 \in E(G_1)$ and $b_1b_2\in E(G_2)$.

The next result establishes the relationship between split orientations and induced matchings in a suitably defined product graph. 

\begin{lemma}\label{lem:bijection}\app
    For a given graph $G$, the family of its possible split orientations is bijective with the family of possible induced matchings of $G \times e$, where $e$ is a single edge.
\end{lemma}

Next, we show how to employ split orientation to create a ``choice gadget'' (for the hardness reduction), but we need to handle one set of vertices separately, the leaves.
Notice that one can prevent a vertex being a root of a tree in two different ways.
Either connect its neighbor to a vertex that is guaranteed to be a neighbor of a root, or connect two of its neighbors with a 0-edge by \Cref{obs:niceEFXObservations} (2). % \Cref{obs:donot_root_when_vertex_has_zeroes}.
We may approach leaves of a tree in two different ways, either we prevent them from being roots by the former method (as the latter needs two distinct neighbors), which also forbids all their second neighbors, or we force them all to be possible roots -- we do the latter.

\begin{figure}[ht]
    \centering
    \hfill
    \begin{subfigure}{0.48\textwidth}
        \centering
        \includegraphics[page=1,scale=1.1]{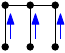}
        \subcaption{the maximum leafed split orientation}%
        \label{fig:t112a}
    \end{subfigure}
    \hfill
    \begin{subfigure}{0.48\textwidth}
        \centering
        \includegraphics[page=2,scale=1.1]{t112.pdf}
        \subcaption{the maximum split orientation}%
        \label{fig:t112b}
    \end{subfigure}
    \hfill
    \caption{
        A graph where its maximal leafed split orientation is smaller than its maximum split orientation.
        Split orientations are shown as arrows.
    }%
    \label{fig:t112}
\end{figure}

The idea behind split orientation is to represent roots as the arc tails and their ``private'' neighborhoods as arc head.
To implement the above approach to tackling leaves we aim to ensure that the split orientation contains for each leaf an arc that has its tail in the leaf. That perspective leads us to a ``leafed"-version of a split orientation: %, as defined below.%\begin{definition}
    A \emph{leafed split orientation} is a split orientation $A$ of $T$ such that $\{(u,v) \mid u,v \in V(T), uv \in E(T), \text{$u$ is a leaf of\/ $T$}\} \subseteq A$.
%\end{definition}

We note that this may not yield a split orientation of maximum cardinality, as shown in \Cref{fig:t112}.
Any tree that is not a core has no leafed split orientation because including both leaves as tails of arcs makes the two arcs break the split orientation property.
On the other hand, we can now show that any core tree has a leafed split orientation, and it can found in linear time.

\begin{lemma}\label{lem:max_leaf_so}
    A core with maximum leafed split orientation of size $k \ge 3$ has at least $k+1$ and at most $5k-1$ vertices.
    For a fixed core $T$ we can find its maximum leafed split orientation in linear time.
\end{lemma}

\begin{proof} We break the analysis into smaller claims which we will prove separately. We begin with the following.

\begin{claim}\label{clm:core-has-split-orientation}\app
    A core always has some leafed split orientation.
\end{claim}

    To compute $A$ a leafed split orientation of maximum cardinality, we consider the bijection $f \colon E(T \times e) \to A(T)$ between arcs and the edges of $T \times e$ from \Cref{lem:bijection}.
    Recall that the elements of $T \times e$ are denoted by $u_0$ and $u_1$, where $u \in V(T)$. Observe that $T \times e$ consists of two copies $T_0,T_1$ of $T$ if and only if $T$ is bipartite, which is true for our case.
    More precisely, let us fix an arbitrary vertex $r \in V(T)$ and let $d_r(u)$ for $u \in V(T)$ denote parity of the distance from $r$ to $u$.
    We define $T_0 = (T \times e)[\{u_q \mid u \in V(T), q \equiv d_r(u) \mod 2\}]$ and $T_1 = (T \times e)[\{u_q \mid u \in V(T), q \equiv d_r(u) + 1 \mod 2\}]$. That is, for each $u \in V(T)$ we have $u_0 \in T_0$ and $u_1 \in T_1$ if the distance from $r$ to $u$ is even, and $u_0 \in T_1$ and $u_1 \in T_0$ if the distance from $r$ to $u$ is odd. Let $L_1 = \{u \mid u_0 \in T_0\}$ and $L_2 = \{u \mid u_1 \in T_0\}$. We note that $L_1 \uplus L_2$ form a partition of the leaves of $T$. 
    
    \begin{claim}\label{clm:leaves-are-at-distance-4}
    \app 
        Each pair of leaves from the same part $L_1$ or $L_2$ are at a distance at least 4.
    \end{claim}
    
    Now we are ready to bound the size of the maximum leafed split orientation.
    Towards that we invoke \Cref{lem:leafedmimisniceandfast} with the partition $L_1\uplus L_2$ of the leaves of $T$ which computes maximum induced matchings $R_1$ and $R_2$ that saturate $L_1$ and $L_2$, respectively in linear time. For each $i\in [2]$, let $g_i$ be a map from vertices of $T$ to respective vertices of $T_i$. 

    \begin{claim}\label{clm:maximum-leafed-orientation}
    \app
        Function  $f(g_1(R_1) \cup g_2(R_2))$ yields a maximum leafed split orientation of $T$; and when $|A|\ge 3$, we have $|A|+1 \leq |T| \leq 5|A|-1$.
    \end{claim}

    The approach we used to devise $A$ uses simple bijections and a linear procedure of \Cref{lem:leafedmimisniceandfast} so we can find $A$ in linear time. This completes the proof of the lemma.
\end{proof}

We can now describe the structure within a core whose existence leads to intractability. 

\begin{lemma}\label{lem:whard_options}
    A core $T$ with leafed split orientation $A$ of size at least $k$ can have added 0-edges in such a way that a solution to \EFXO has exactly $k$ root vertices $R=r_1,\dots,r_k$ such that no other vertex can be a root and each $r_i$ has at least one neighbor $s_i$ that is not a neighbor of any other root.
\end{lemma}
\begin{proof}
    Let us set $r_i,s_i$ to be the endpoints of the $i$-th arc of the split orientation, i.e., $(r_i,s_i) \in A$.
    To prevent non-roots from being selected for each $u \in V(T) \setminus R$ add a 0-edge between two of its neighbors.
    As $A$ is a leafed split orientation $R$ covers all the leaves of $T$ so vertices in $V(T) \setminus R$ are non-leaf vertices and the 0-edge can be added between two distinct neighbors.
    Adding such an edge prevents exactly $u$ from being selected as a root.
    Last observe that due to the properties of split orientation we have $r_is_j \not\in E(G)$ for each $i \ne j$, proving that $s_i$ is neighboring only the root $r_i$ for each $i \in [k]$.
\end{proof}

By \emph{editing 0-edges} of $G$ to obtain $G'$ we mean that $G'_1=G_1$ but $G'_0$ may differ from $G_0$ arbitrarily.
A set $\mathcal C$ of \EFXO instances is \emph{closed under editing 0-edges} if for every $C \in \mathcal C$ the set $\mathcal C$ contains every instance $C'$ that can be obtained from $C$ by editing 0-edges.
Consequently, we are now well-positioned to present the hardness result. %using the gadget developed so far. 

% \complexcomponents*
\begin{restatable}{theorem}{complexcomponents}\label{thm:wone_complex_components}\app
    For every $k$, let $\mathcal C_k$ be a set of \EFXO instances that is closed under editing 0-edges and for any $n_0 \in \mathbb N$ contains an instance in $\mathcal C_k$ such that $G_1$ contains $k$ trees with core size at least $5 n_0-1$. Then, 
    \EFXO on instances from $\mathcal C_k$ is 
    $W[1]$-hard 
parameterized by $k$.
\end{restatable}

\section{Concluding Remarks}

We conclude by noting that as pointed out in \cite{christodoulou_fair_2023}, major problems in algorithmic game theory, such as the complexity of Nash equilbria \cite{ChenDengTeng09/journal,Daskalakis09/journal} and truthful scheduling (Nisan-Ronen
conjecture)~\cite{NisanRonen/journal,Christodoulou23,Christodoulou22/FOCS} did not lose their characteristic complexity when studied in the restricted graph setting. If anything, the valuable insights from the graphical setting may have helped in finally resolving them in their full generality. We posit in the same vein that using insights and ideas from parameterized graph algorithms can only flesh out the nuances of the graphical setting even further and can only take us closer to the general problem. 

%%%%%%%%%%%%%%%%%%%%%%%%%%%%%%%%%%%%%%%%%%%%%%%%%%%%%%%%%%%%%%%%%%%%%%%%%%%%%%%%
\newpage
\bibliographystyle{ACM-Reference-Format} % EC style
\bibliography{main}

%%%%%%%%%%%%%%%%%%%%%%%%%%%%%%%%%%%%%%%%%%%%%%%%%%%%%%%%%%%%%%%%%%%%%%%%%%%%%%%%

\newpage

%%%%%%%%%%%%%%%%%%%%%%%%%%%%%%%%%%%%%%%%%%%%%%%%%%%%%%%%%%%%%%%%%%%%%%%%%%%%%%%%

\appendix

\section{Appendix}

\subsection*{Additional related work}\label{additional-rw}
The subject of fair division is vast. We point the reader to the surveys \cite{mishra2023fair/survey} and \cite{amanatidis2023fair/survey}. In addition to the papers referred earlier, we can point to a few others that have made significant contributions to the study of EFX.

\paragraph{Related results on graphical EFX allocations.} Amanatidis et al.~\shortcite{ApproxEFX-A/EC24} study approximate graphical EFX allocations showing that a 2/3-\EFX allocation always exists and can be efficiently computed for agents with additive valuation functions in Christodoulou's multi-graph setting. Zhou et al.~\shortcite{10.24963/ijcai.2024/338} show that \EFXp orientation always exists and can be found in polynomial-time, drawing a contrast with \EFX orientation which may not always exist,~\cite{christodoulou_fair_2023}. Moreover, they initiate work on graphical EFX on {\it chores}, also studied by \cite{HsuKing24/EFX-O-chores}; and {\it mixed manna}, which is a generalization of both. %The chore 

\paragraph{Beyond \EFX orientation.} The graphical setting has been extended to study EF1 and EF~\cite{MisraSethia/ADT24} as well. Deligkas et al.~\shortcite{Deligkas/EFX-EF1-orientations} studies primarily the question of EF1 orientation for hypergraphs and explores the parameterized complexity of EFX orientation. They also strengthen the hardness result of \cite{christodoulou_fair_2023} by showing that the \EFX orientation problem NP-complete even on graphs with constant-sized vertex cover and additive symmetric valuations. Furthermore, for multi-graphs, the NP-hardness extends to instances with only a constant number of agents, 8. Chandramouleeswaran et al.~\shortcite{Chandramouleeswaran/FairDivisionVariableSetting} also study EF1 allocation in the multi-graph setting in relation to problem of EF1 restoration.

\paragraph{Charity}
Caragiannis et al. \shortcite{Caragiannis19/EFX-Charity} initiated the research of finding desirable EFX allocations
that satisfy certain properties, such as maximizing the Nash social welfare, where some goods remain unallocated, i.e the {\it charity} bundle. There has been several follow-up work that have reduced the number of items in the charity bundle, ~\cite{Chaudhury21a},\cite{Chaudhury21b} \cite{Akrami22a},\cite{Berendsohn/MFCS17}, culminating in Berger et al.~\shortcite{BergerCohenFeldmanFiat22} showing that for four agents, 
 giving at most one good to charity is enough.

\paragraph{Chores}Work on EFX allocation for chores has caught steam. Recently, \cite{Christoforidis24/EFX-chores} showed that an EFX allocation for chores need not exist under general cost functions. They present a construction
with three agents in which no bounded approximation exists. Moreover, deciding 
if an EFX allocation exists for an instance with three agents and superadditive costs is NP-complete.

\paragraph{Intractibility} Goldberg et al.\shortcite{Goldberg/EFX-PLS} study the intractability of EFX with just two agents, and find that computing an EFX allocation for two identical agents with submodular valuations is PLS-hard. Essentially they show that the problem becomes intractable as the valuations become more general. 

\paragraph{Graphical fair division} This is a fast growing topic in which several different fairness criterion, be it comparison based or share-based have been studied: conflict-free fair division~\cite{Chiarelli2022fair}~\cite{Gupta/journal/Budget-cinflict-free}, connected fair division~\cite{deligkas2021parameterized}, compact fair division~\cite{madathil2023fair}, gerrymandering on planar graphs~\cite{dippel2023gerrymandering} are just a few. Graph based EFX allocations is an addition to this body of work.

\subsection*{Proof of \Cref{obs:strong_envy_conditions}}

    Consider an orientation of $G=(V,E)$ and let $\{X_w\mid w\in V\}$ denote the bundles derived from the orientation. 
 
    Let us first suppose that $u$ strongly envies $v$ and argue that the three conditions hold. 
    Suppose for a contradiction that $(u,v) \not\in E(\OG_1)$. Then either $uv \not\in E(G)$ or $uv \in E(G_0)$ or $(v,u) \in E(\OG_1)$. 
    In the former two cases we have $w_u(X_v)=0$ so $u$ cannot strongly envy $v$'s bundle, therefore, $u$ cannot strongly envy $v$ if they do not share a 1-edge.
    In the last case, $w_u(X_u) \ge 1$ and as maximum $w_u(X_v)$ is $1$ we know that $u$ cannot strongly envy $v$'s bundle. This completes the argument for the first condition.     In fact, the same reasoning applies for the second condition as well. Precisely, when agent $u$ gets an item other than $uv$ that they value as 1, it follows that $w_u(X_u) \ge 1$ while the maximum $w_u(X_v)$ is $1$.  Thus, $u$ cannot strongly envy $v$'s bundle.
    For the final condition, if $v$ gets no item other than $uv$, then $u$ cannot strongly envy $v$ since strong envy requires $w_u(X_v \setminus \{g\}) > w_u(X_u)$ for some item $g$, but as $X_v=\{g\}$ we would need $w_u(\emptyset)>0$, which is a contradiction.

    Now, consider the converse direction and suppose that the three conditions hold.
    Due to the second condition, $w_u(X_u) = 0$.
    By the first condition, $uv \in X_v$ and that $w_u(uv) = 1$ and the third condition ensures that there is some item $vw$ distinct from $uv$ such that $vw \in X_v$.
    Together these imply $w_u(X_v \setminus \{vw\}) \geq 1 > 0 = w_u(X_u)$. So, $u$ strongly envies $v$. \qed

\subsection*{Proof of \Cref{lem:simple_reductions}}
    For \cref{it:isolated}, since $G'$ is a subgraph of $G$, an EFX-orientation for $G$ can be used to find an EFX-orientation for $G'$ by simply removing the edges not in $G'$.
    Removing edges cannot cause strong envy through \Cref{obs:strong_envy_conditions} (1) or (3) and since $u$ is isolated in $G_1$, all of these are zero edges so they will not affect the application of \Cref{obs:strong_envy_conditions} (2).
    Conversely, given an EFX-orientation for $G'$ we can obtain an EFX-orientation for $G$ by orienting all edges incident to $u$, towards $u$.
    Since $u$ is isolated in $G_1$, all of these edges are zero edges so there is no strong envy involving $u$ by \Cref{obs:strong_envy_conditions} (1).  Finally every other agent has exactly the same bundle in our new orientation as in the original orientation of $G'$ so there is no other strong envy.

    For \cref{it:component}, the forward direction is similar to \cref{it:isolated}:
    all deleted edges that are incident to vertices in $G'$ are 0-edges, so will not cause strong envy between their endpoints.
    For the backward direction we construct an orientation for the edges in $G$ but not $G'$ by first orienting the 1-edges of $C$ in an (arbitrary) direction around the cycle to create a directed cycle.
    We can now find a spanning tree of $H$ which contains all the edges of $C$ except one and orient the edges of this tree away from $C$.
    Then we orient all remaining unoriented edges in $H$ arbitrarily,  and orient all 0-edges incident to $V(H)$ in $G$ towards the vertices of $H$. A 0-edge between two vertices in $V(H)$ is oriented arbitrarily.
    As all agents $V(H)$ get at least one 1-item they cannot strongly envy any other agents irrespective of the orientation of the rest of $G'$ by \Cref{obs:strong_envy_conditions} (2).
    Since $H$ is a connected component, for every other agent $v$ not in $V(H)$, there is no 1-edge incident to both $v$ and a vertex in $H$.   So there is no strong envy by \Cref{obs:strong_envy_conditions}~\ref{envy:one_edge}. \qed

\subsection*{Missing details of the proof of \Cref{thm:tractable_on_pfive_free}}
    After the above preprocessing each component of $G_1$ is a tree $T_i$ of diameter at most 3.
    This diameter is realized between two leaves, we have three cases.
    If diameter is 1, then the tree is a single edge ($P_2$) which trivially has at most 2 states.
    If diameter is 2, then we have a star.
    As any star on 4 or more vertices gets reduced to a smaller one by \Cref{lem:core_procedure} we know that this star has 3 vertices, i.e. it is a $P_3$.
    Two of its leaves share a parent so one dominates the other, implying that there are at most 2 states.
    If diameter is 3, then we have two non-leaf vertices $u_1$ and $u_2$ that each cannot be connected (aside from each other) to anything else but leaves.
    By \Cref{lem:core_procedure} each must have at most one leaf, so the tree is a $P_4$.
    Let $v_1$ be a leaf of $u_1$ and let $v_2$ be a leaf of $u_2$.
    Note that $N_1(v_1) \subseteq N_1(u_2)$ and $N_1(v_2) \subseteq N_1(u_1)$ so the two leaves dominate the other vertices, and the number of states is at most two.

    Notice that the states devised for the trees have disjoint neighborhoods so we can reduce each tree $T_i$ of $G_1$ to a single edge in the following way.
    For each $i \in [\ell]$ such that $T_i$ has diameter more than 1 we run the following procedure.
    Let $u$ and $v$ be the two vertices witnessing the two states of $T_i$, let $Z_u$ and $Z_v$ be the set of $G_0$ edges with endpoints in $N_1(u)$ and $N_1(v)$, respectively.
    We remove $V(T_i)$ from the graph, then add $u$ and $v$ as new vertices, add a 1-edge $uv$, we take the set $Z_u$ and put them back to the graph while connecting them to $v$ instead of the neighbors $N_1(u)$, similarly, we put back the edges of $Z_v$ and connect them to $u$ instead of $N_1(v)$.
    As the neighborhoods of the new states are the same as for the states in the original tree, we see that the two new states are equivalent to the prior two states.

    By repeating the above procedure we end up with a matching in $G_1$.
    For each edge of $e_1,\dots,e_n \in G_1$ let us denote endpoints of $e_i$ by $x_i$ and $\neg x_i$, this represents where the edge gets rooted.
    Let us denote now the vertices and variables interchangeably, by $\neg u$ we mean the other vertex in the 1-edge that contains $u$.
    Next, for each zero edge $uv = e \in E(G_0)$ create a clause $C_e = (u \lor v)$, this represent the infeasibility to root respective trees in $\neg u$ and $\neg v$ at the same time.
    Conjunction of all these clauses creates the final 2-CNF formula $\phi = \bigwedge_{e \in E(G_0)} C_e$.

    To show the equivalence, consider a satisfying assignment for the 2-CNF formula $\phi$.
    Let us root each tree $G_1$ in the vertex denoted by the variable that resolves to true.
    This roots every tree of $G_1$ and the union of root neighborhoods is independent as the violating 0-edge $uv$ is represented in $\phi$ by $(u \lor v)$, which would not have been satisfied.
    In the other direction, consider a rooting that corresponds to a solution of the \EFXO instance.
    Then setting the variable $x$ to be true if and only if tree of $x$ was rooted in $x$ results in the fomula to be satisfied as every clause $(u \lor v)$ represents an edge $uv \in E(G_0)$ and it being not satisfied would mean we picked $\neg u$ and $\neg v$ which would violate the characterization of a \EFXO solution from \Cref{prop:structure}. %\todo{which characterization?}
    % Assuming a reasonable representation of the input graph we note 
    Note that this reduction takes linear time and as 2-SAT has a linear-time algorithm, see e.g.~\cite{ASPVALL1979121}, so we can solve the \EFXO instance in linear time.
\qed

\subsection*{Proof of \Cref{cor:alwaysnice}}
    Each tree in $G_1$ has fewer edges equal to the number of vertices minus one so we know that in any orientation at least one agent gets no 1-item.   There could be multiple such agents, so for each tree, pick one of these agents arbitrarily as the root. Then, reorient each tree (if needed) so that all 1-edges are oriented away from the chosen root. For each agent aside from the roots,  this ensures that the number of 1-items they are allocated is precisely one.
    By \Cref{obs:strong_envy_conditions}~\ref{envy:value_one} the only vertices that could experience strong envy after the tree reorientation are the roots.
    However, prior to the reorientation all the 1-neighbors of any root already had received at least one 1-item, and since there was no strong envy earlier, by \Cref{obs:strong_envy_conditions}~\ref{envy:other_item} these neighbors had been allocated no other items.    So, the reorientation could not have introduced strong envy in the roots.  Hence, we may assume that the result is an EFX-orientation that contains exactly one agent without a 1-item in each tree, i.e., it is a nice EFX-orientation. \qed

    \subsection*{Proof of \Cref{lem:corollaryEOCT}}
    Given $e$ such that $G-e$ is a bipartite graph and compute a bipartition $(X,Y)$ of $G-e$. This can be done in linear time by just computing the spanning forest of $G-e$ and taking the unique (up to flipping sides) bipartition of each tree. 
    % Note that the sets $X$ and $Y$ can be
    
    Moreover, notice that $e$ must have both endpoints in $X$ or in $Y$ as otherwise $G$ is already bipartite. So, assume without loss of generality that both endpoints are in $X$. If $w(e)=1$, then set $A:=X$ and $B:=Y$. Otherwise, set $B:=X$ and $A:=Y$. It is straightforward to check that the bipartition $(A,B)$ satisfies the premise of \Cref{thm:Edge-OCT1-general-condition}. 
\qed

\subsection*{Missing details of \Cref{thm:edgeoct2}}

\subsubsection*{Proof of \Cref{lem:reductionindependence}}
    Assume $u_2$ or $u_4$ get no 1-items, say $u_2$ (argument for $u_4$ is identical).
    Then $u_3$ must get at most one item, otherwise $u_2$ would strongly envy $u_3$.
    Hence $u_3 x_{u,2}$ must be oriented away from $u_3$.
    Since the orientation is good, $x_{u,3}$ is assigned a 0-edge so neither $x_{u,2}$ nor $x_{u,4}$ can be roots of $Q_u$.
    Now $x_{u,2}$ got a 0-edge so if $x_{u,1}$ or $x_{u,3}$ were roots they would be strongly envious of $x_{u,2}$.
    Therefore, $Q_u$ must be rooted at $x_{u,5}$.
    Again, $x_{u,4}$ must get only one item, $x_{u,4} v_3$ is oriented away from $x_{u,4}$, so $v_3$ gets a 0-edge.
    Now if $v_2$ or $v_4$ were chosen as roots of their path they would be strongly envious of $v_3$. \qed

\subsubsection*{Proof of \Cref{clm:goodiffmis}}
    We first assume we have a solution $S$ to our instance of MIS.
    For each $i \in [k]$ let $\{v\} = S \cap V_i$ we choose $v_2$ as the root for $P_i$ in our solution to \EFXO.
    Additionally, for each $v \in V(G')$ we choose $x_{v,5}$ as a root for $Q_v$ if $v \in S$ and $x_{v,1}$ if $v \notin S$.
    Since $S$ is multicolored we choose exactly one root from each $P_i$.
    Since $S$ is an independent set for any edge $uv$ the solution does not contain both $u$ and $v$.
    If $u \in S$ we rooted in $x_{u,5}$ but not in $v_2$ and if $v \in S$ we rooted in $v_2$ and in $x_{u,1}$ because $u \notin S$.
    In both cases the edge $x_{u,4} v_3$ will not cause strong envy.

    Conversely, we assume we have a good EFX-orientation of $G$, and let the set of roots be $S'$.
    For each color $i$ we find the vertex $u$ that satisfies $u_2 \in S'$ or $u_4 \in S'$ and add it to our solution $S$ to MIS.
    Every EFX-orientation has at least one root in each component, and in particular in each $P_i$ so $S$ is colorful.
    Suppose that there exists an edge $uv \in E(G')$ such that $u \in S$ and $v \in S$.
    By construction of $S$ this would mean that $x \in \{u_2,u_4\}$ and $y \in \{v_2,v_4\}$ were chosen as roots.
    But this contradicts \Cref{lem:reductionindependence} for the edge $xy$.
    Therefore, $S$ is an independent set. \qed

\subsection*{Proof of \Cref{red:zero_degrees}}
    We will use the characterization of EFX-orientation from \Cref{prop:structure}.
    Suppose we have applied the operation on $u$ in $G$ to obtain $G'$.
    The operation creates several $P_3$s in $G'_1$.
    Let $L$ be the set of leaves in these $P_3$s and let $M$ denote the  set of middle vertices of the $P_3$s.
    Assuming $G$ has an EFX-orientation then either $u$ is not a neighbor of root or all vertices in $N_0(u)$ are not neighbors of roots.
    In the former case, we root all $P_3$s in vertices of $L$ (chosen arbitrarily), in the latter case we root all $P_3$s in $M$.
    This orientation in $G'$ is EFX because the tree $B$ is bipartite with parts $L \cup \{u\}$ and $M \cup N_0(u)$, and as we root in only one of these sets we may not get two neighbors of roots that are adjacent in $G_0$.

    For the converse, suppose we find an EFX-orientation in $G'$, consider one $P_3$ consisting of leaves $\ell_1$, $\ell_2$, and a vertex $v$.
    If $\ell_1$ or $\ell_2$ is a root of the $P_3$ then $v$ is a neighbor of a root.
    Hence, its parent within $B$ cannot be a neighbor of a root, so that $P_3$ must be also rooted in its leaf.
    Continuing this chain of implications we conclude that if in $B$ any $P_3$ is rooted in $L$ then $u'$ is a neighbor of a root and $v$ cannot be a neighbor of a root.
    If any vertex in $N_0(u)$ is a neighbor of a root then the $P_3$ that contains its 0-neighbor must be rooted in $L$, and therefore, $u$ cannot be a neighbor of a root.
    In $G$ we orient all edges from $N_0(u)$ to $u$ in the same direction as the edge from $u'$ to $u$ in $G'$ was oriented to get an EFX-orientation. \qed

\subsection*{Proof of \Cref{thm:3satnpc}}
Given an instance of \PMSAT we apply \Cref{lem:efxosmallnpc} to obtain an instance of {\EFXO} and then apply \Cref{red:zero_degrees} exhaustively.
More precisely, take the result of \Cref{lem:efxosmallnpc} and observe that for each vertex $u \in V(G)$ its 0-edges form a connected subsequence $P$ of the circular order of curves that represent edges incident to $u$ in the planar embedding.
When replacing these 0-edges with a tree $T$ during \Cref{red:zero_degrees} we embed $T$ such that its leaves are ordered according to $P$ so that the result of the reduction is planar.
This introduces intermediate paths $P_3$ on 1-edges and after performing the reduction of 0-degrees on all vertices, $G_0$ becomes a matching. This implies that the resulting graph has maximum degree 3, concluding the proof of \Cref{thm:3satnpc}.

\subsection*{Proof of \Cref{lem:core_procedure}}
Let $u$ and $v$ be a pair of leaves that witness that a tree $T$ in $G_1$ on at least 4 vertices is not a core, i.e., leaves $u$ and $v$ share a neighbor $p$.
Let $q$ be another vertex adjacent to $p$ distinct from $u$ and $v$.
Let \emph{merge} of $u$ and $v$ be the operation performed as follows: if there is a zero edge $uv$ reconnect it to $ur$,  contract $v$ into $u$, and remove multiple edges to obtain a simple graph. 

Observe that applying the merge operation exhaustively to a tree makes it a core. 
The operation removes a leaf of $T$ (in $G_1$) and reconnects some 0-edges.
We can implement the exhaustive application of this procedure by traversing a tree $T$ and for each non-leaf vertex take all its adjacent leaves and merge all of them to a single leaf, choosing $q$ to be the same other vertex for all the merges (for a special case where we have a star we merge all but one leaf which is selected to be $q$, ending up with $P_3$ after the merges).
Observe that  exhaustively performing the merging operation is a linear-time procedure since for each leaf we perform constant number of operations and every 0-edge is reconnected at most once for each of its endpoints.

Let $G$ be the initial graph and $G'$ be graph $G$ after merging $u$ with $v$.
These vertices are in $G_1$ tree, have a common neighbor $p$.
We argue correctness by showing that $G$ and $G'$ are equivalent instances.

Assume $G$ has a rooting $R$.
Let us create $R'$ that is the same as $R$ but if $v \in R$ then we add $u$ to $R'$.
Towards a contradiction, assume $R'$ is not a rooting of $G'$, then there must be a 0-edge that has both endpoints adjacent to roots.
This cannot be the case for edges that are not adjacent to $u$ in $G'$ as those were not altered and would be adjacent to roots in $G'$ as well.
For the 0-edges adjacent to $u$ we have two cases.
Say $uq$ is the violating edge, this means $p$ is the root, but $uq$ either was in $G$ already, or $vq$ was in $G$, or it was created due to a 0-edge $uv$ which also has both endpoints neighboring $p$, a contradiction.
Similarly, for other edges $uw$ where $w\in T$ the edge must have either been in $G$ already or we had $vw$ in $G$, both still contradict $p$ being a root in $G$.
For the last case, say $uw$ is a violating edge where $w$ is in a tree different from $T$, then $G$ must contain either $uw$ or $vw$, both endpoints are still neighboring roots and result in a contradiction.

For the other direction, assume $G'$ has a rooting $R'$.
We keep the same rooting $R=R'$ for $G$ and claim it has no violating edges.
This is simply because all vertices in $V(G')$ have their neighbors through 0-edges in $G$ as a subsets of the neighbors through 0-edges in $G'$.
Hence, any violating edge in $G$ is also a violating edge in $G'$, proving the claim. \qed

\subsection*{A brief primer on MSO and Courcelle's Theorem}
We consider \emph{Monadic Second Order} (MSO) logic on edge- and vertex-labeled
graphs in terms of their incidence structure, whose universe contains vertices and
edges; the incidence between vertices and edges is represented by a
binary relation. We assume an infinite supply of \emph{individual
variables} $x,x_1,x_2,\dots$ and of \emph{set variables}
$X,X_1,X_2,\dots$. The \emph{atomic formulas} are 
$V(x)$ (``$x$ is a vertex''), $E(y)$ (``$y$ is an edge''), $I(x,y)$ (``vertex $x$
is incident with edge $y$''), $x=y$ (equality),
$P_a(x)$ (``vertex or edge $x$ has label $a$''), and $X(x)$ (``vertex or
edge $x$ is an element of set $X$'').  \emph{MSO formulas} are built up
from atomic formulas using the usual Boolean connectives
$(\lnot,\land,\lor,\Rightarrow,\Leftrightarrow)$, quantification over
individual variables ($\forall x$, $\exists x$), and quantification over
set variables ($\forall X$, $\exists X$).

\emph{Free and bound variables} of a formula are defined in the usual way. To indicate that the set of free individual variables of formula $\Phi$
is $\{x_1, \dots, x_\ell\}$ and the set of free set variables of formula $\Phi$
is $\{X_1, \dots, X_q\}$ we write $\Phi(x_1,\ldots, x_\ell, X_1,
\dots, X_q)$. If $G$ is a graph, $z_1,\ldots, z_\ell\in V(G)\cup E(G)$ and $S_1, \dots, S_q
\subseteq V(G)\cup E(G)$ we write $G \models \Phi(z_1,\ldots, z_\ell, S_1, \dots, S_q)$ to denote that
$\Phi$ holds in $G$ if the variables $x_i$ are interpreted by (i.e., set to be equal to) the vertices or edges $z_i$, for $i\in [\ell]$, and the variables $X_i$ are interpreted by the sets
$S_i$, for $i \in [q]$. 

\subsection*{Proof of \Cref{lem:mim_size_relation}}
As the tree contains at least $m$ independent edges which all need to be connected to each other by another vertex its size is at least $2m+1$, this is a tight lower bound because once-subdivided $m$-leaf stars attain this value, see \Cref{fig:matching_examples_lowerbound}.

    On the other hand, consider a graph on at least 5 vertices and let $u$ be one of its most eccentric leaves, i.e the leaf with the highest eccentricity. Let $v$ be the vertex adjacent to $u$, and let $w$ be the vertex other than $u$ adjacent to $v$.
    Note that $w$ cannot be a leaf as we work with a core, and the choice for $w$ is unique because otherwise $u$ would not be one of the most eccentric vertices.
    We claim that there exists a maximum induced matching that contains the $uv$ edge.
    Suppose this is not the case, then if no vertex $u,v,w$ is in the induced matching then we could add $uv$ to the matching, contradicting its maximality.
    Otherwise, we have a matching edge $vw$ or a matching edge that has one endpoint in $w$.
    In either case, we remove that edge from the matching, and add $uv$ instead.
    The result is still an induced matching since there are no other edges adjacent to $uv$ and size of the matching did not change, completing the claim proof.

    Suppose we apply the procedure described above to find a maximal induced matching of the tree as follows.
    Find $u,v,w$ as described, then set $uv$ to be a matching edge and remove vertices $u,v,w$ from the tree.
    We then invoke \Cref{lem:core_procedure} to reduce the remaining trees to cores so that we can use induction.

    For the base case we consider all graphs with at most 5 vertices and conclude:
        1 vertex tree has no matching edge,
        trees with 2, 3, and 4 vertices have exactly one matching edge,
        and 5 vertex tree is only $P_5$ which has 2 induced matching edges.
    Now we aim to count the maximum number of removed vertices due to the $uv$ matching edge.
    Clearly in each step we removed three vertices $u,v,w$, but additional vertices are removed by \Cref{lem:core_procedure} or by becoming base case components.
    Note that the processing can remove only leaves that were created by removing $u,v,w$, hence, vertices of degree 2 that were incident to $w$.
    We claim that the procedure removes at most one leaf.
    Towards a contradiction, assume that the procedure removed two leaves, let us name them $r_1$ and $r_2$.
    Vertex $r_1$ was removed because there is a twin leaf $r'_1$ and both $r_1$ and $r'_1$ have a common neighbor $p_1$.
    Similarly, $r_2$ was removed because there is a twin leaf $r'_2$ and both have a common neighbor $p_2$.
    Note that both $r'_1$ and $r'_2$ are at distance 3 from $w$, hence $u$ could not have been the most eccentric vertex, a contradiction that proves our claim.
    We count this one vertex reduced by the preprocessing towards $uv$.
    For the vertices that became base cases, by the above argument we cannot have two components with vertices that are further from $w$ than $u$.
    Hence, every component of $G \setminus \{u,v,w\}$ except one must contain at most 2 vertices.
    Note that if the component contains 2 vertices then it can have a matching edge, so removal of those vertices can be counted towards their matching edge and not towards $uv$.
    The bad case is when the component contains only 1 vertex; we remove the vertex and count it towards $uv$.

    \begin{figure}[ht]
        \centering
        \hfill
        \begin{subfigure}{0.38\textwidth}
            \centering
            \includegraphics[page=2,scale=1.1]{matching_examples.pdf}
            \subcaption{Once subdivided star attains the minimum number of vertices for (11) $m=5$.}
            \label{fig:matching_examples_lowerbound}
        \end{subfigure}
        \hfill
        \begin{subfigure}{0.58\textwidth}
            \centering
            \includegraphics[page=1,scale=1.1]{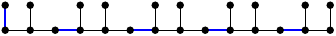}
            \subcaption{This path with pendant leaves attains the maximum number of vertices (24) for $m=5$.}
            \label{fig:matching_examples_upperbound}
        \end{subfigure}
        \hfill
        \caption{Example graphs that attain extremes with respect to their maximum induced matchings of size $m$.}%
        \label{fig:matching_examples}
    \end{figure}

    As the above procedure removes at most 5 vertices at each step.
    For the base cases we observe that to reach 1 vertex tree we ought to have $P_5$ just before that.
    So for any non-trivial tree the worst base case is the graph on 4 vertices which has a maximum induced matching of size 1.
    All this gives us an upper bound of $5m-1$ on the maximum number of vertices in a tree with maximum induced matching $m$ -- this is tight as the this value is attained on a path $P_{3m+1}$ with vertices $p_1,\dots,p_{3m-1}$ that has $2m$ leaves attached to vertices $p_{3i-2},p_{3i-1}$ for all $i \in [m]$, see \Cref{fig:matching_examples_upperbound}. \qed

    \subsection*{Details of the proof of \Cref{lem:leafedmimisniceandfast}}

    First, observe that if there are $u,v \in L_i$ such that their distance is at most $3$ then any set of edges that contains both $u$ and $v$ cannot be an induced matching as the neighbors of $u$ and $v$ are adjacent to each other.

    On the other hand, if each pair in $L_i$ is at distance at least $4$ then we can take the edges incident to $L$ and add further edges to create a maximum leafed induced matching $R$. 

    We first design a dynamic programming algorithm to compute the leafed matchings and then show the lower bound on its size.
    We root $T$ at an arbitrarily chosen non-leaf

    Let us fix $i \in [2]$.
    We aim to define a table $\D \colon V(T) \times \{\no,\btt,\tpp\} \to \mathbb N \cup \{-\infty\}$ in which we compute the size of maximum leafed induced matchings. More precisely, for the subtree $T_t$ of $T$ rooted at the vertex $t \in V(T)$, the entry $\D(t,j)$ contains the maximum induced matching in $T_t$ that contains all the vertices in $L_i \cap V(T_t)$ and using the following conditions based on the value of $j \in \{\no,\btt,\tpp\}$:
    \begin{itemize}
        \item $\D(t,\no)$ -- Vertex $t$ is not contained in a matching edge.
        \item $\D(t,\btt)$ -- Vertex $t$ is contained in a matching edge as the vertex that is further from the root $r$.
        \item $\D(t,\tpp)$ -- Vertex $t$ is contained in a matching edge as the vertex that is closer to the root $r$.
    \end{itemize}
    The computation of $\D$ differs for leaves and non-leaf vertices.
    \begin{description}%[wide=0pt]
    \item[When $t$ is a leaf:] We set $\D(t,\btt)=1$, $\D(t,\tpp)=-\infty$, and last if $t \in L_i$ we have $\D(t,\no)=-\infty$, otherwise $\D(t,\no)=0$.
    \item[When $t$ is a non-leaf vertex:] The computation of $\D$ follows from the results for its children set $t_1,\dots,t_q$ as follows:
    \begin{equation*}
        \begin{split}
            \D(t,\no)  &= \sum_{j=0}^q \max \{ \D(t_j,\no), \D(t_j,\tpp) \} \\
            \D(t,\btt) &= 1 + \sum_{j=0}^q \D(t_j,\no) \\
            \D(t,\tpp) &= \max_{j=0}^q \big( \D(t_j,\btt)-\D(t_j,\no) \big) + \sum_{j=0}^q \D(t_j,\no)
        \end{split}
    \end{equation*}
    \end{description}
    The size of a maximum leafed induced matching of $T$ is retrieved from the root as $\max_{j \in \{\no,\tpp\}} \D(r,j)$.

    We can prove the correctness formally using a standard inductive argument. We note that if we want to add an edge of a matching we need to first have $\btt$ and its parent $\tpp$, adding 1 to the result and requiring all neighboring vertices to have $\no$; taking maximum ensures we keep the best solution with the given assumptions.
    It is not hard to see that in the above procedure each operation can be computed in the complexity given by the number of children of each node, and hence, using a post-order computation over the tree can be implemented in $\mathcal O(|T|)$.

    The above process results in $R_1$ when $i=1$ and $R_2$ for $i=2$.
    Consider the maximum induced matching $M$ of $T$. 

     \begin{claim}\label{clm:size-of-ind-matching-in-L1-and-L2}
  
        We have $|R_1|+|R_2| \geq |M|$.    %\qedhere
    \end{claim}

\begin{proof}
    Let $B_1$ be the set of all edges incident to $L_1$, and $B_2$ be the edges incident to $L_2$.
    Let $Q_1$ be edges of $M$ adjacent to $B_1$, and $Q_2$ be edges adjacent to $B_2$.
    Let $M_1$ and $M_2$ be two leafed induced matchings created as follows for each $i \in \{1,2\}$.
    Set $M_i = M \cup B_i \setminus Q_i$.
    Note that this creates an induced matching because we removed all edges of $M$ that would be incident to the added edges, and additionally all vertices in $L_i$ are mutually at distance at least 4 so edges in $B_i$ are mutually non-adjacent.
    As $M_i$ contains all leaves $L_i$ it is a leafed induced matching.
    Set $R_i$ is the maximum induced matching with leaves $L_i$ so we have $|R_i| \ge |M_i|$.
    We claim that $|M_1|+|M_2| \ge |M|$ and prove this claim by showing that for every $e \in M$ we can associate a unique edge from $M_1 \cup M_2$.
    Every edge $e \in M$ such that $e \in M_1$ or $e \in M_2$, we associate it with the respective copy in $M_1$ or $M_2$.
    The remaining edges were removed in both copies so they are contained in $Q_1 \cap Q_2$.
    Consider an edge $uv \in (Q_1 \cap Q_2)$.
    As $uv \in Q_1$ there is a leaf $\ell_1 \in L_1$ with an edge $\ell_1\ell'_1 \in B_1$ at distance at most $2$ from $u$ or $v$.
    Similarly we have an edge $\ell_2\ell'_2 \in B_2$.
    Recall $T$ is rooted in the vertex $r$, let w.l.o.g. $v$ be the parent of $u$.
    Let $w$ be the parent of $v$.
    For $i \in \{1,2\}$ for the edge $\ell_i\ell'_i$ We have the following two cases.
    \begin{itemize}
        \item If $\ell_i$ is a child of $v$ or $\ell_i$ is a child of $u$ or $\ell_i$ is a grandchild of $u$ -- we associate $uv$ with $\ell_i\ell'_i$.
        \item The remaining case where $\ell_i$ is a child of $w$ can happen only for one value of $i$ as if it happened twice then $w$ would have two leaf children which contradicts $T$ being a core. Hence, for some $i$ the first case always happens.
    \end{itemize}
    For every edge $uv \in M$ we associate it to some $\ell_i\ell'_i$ such that either $\ell_i\ell'_i$ is overlapping $uv$ or $v$ is the parent of $\ell'_i$.
    We conclude with seeing that if two distinct edges were associated to the same $\ell_i\ell'_i$ then they must have been adjacent or overlapping, which contradicts the induced matching property of $M$.
\end{proof}
This completes the proof of the lemma.
\qed

\subsection*{Proof of \Cref{lem:bijection}}
Let us denote endpoints of $e$ as 0 and 1 and product vertex created from $u \in V(G)$ and 0 be $u_0$, similarly let $u_1$ be the vertex for product of $u$ with 1.
We define the bijection as $f \colon E(G \times e) \to A(G)$, where $A(G) = \{(u,v) \mid uv\in E(G)\}$ is the set of all split orientations of the edges in $G$.
Note that $|E(G \times e)| = 2 \cdot |E(G)| = |A(G)|$.

We define $f(u_0v_1) = (u,v)$ for each $u_0v_1 \in E(G \times e)$; as each edge in $E(G \times e)$ is between $V(G) \times 0$ and $V(G) \times 1$ the function $f$ is well defined, and $f$ is a bijection because every element $(u,v) \in A(G)$ has a pre-image, defined by the (undirected) $u_0v_1$. 

We prove that $f$ is a bijection between split orientation and an induced matching in two steps.
First, we assume that $S$ is an induced matching of $G \times e$. Hence, $S\subseteq E(G\times e)$. 
Let $M = f(S)=\{f(e) \mid e \in S\}$. We will prove that $M$ is a split orientation of $G$.

For the sake of contradiction, we assume that some $(u,v),(x,y) \in M$ contradicts $M$ being a split orientation.
Either $uy \in E(G)$ or $xv \in E(G)$, assume without loss of generality that we have the latter.
These arcs can be mapped back to the matching as $f^{-1}((u,v))=u_0v_1$ and $f^{-1}((x,y))=x_0y_1$.
As $xv \in E(G)$ we can get $f^{-1}((x,v))=x_0v_1$ which is an edge in $G \times e$.
But $x_0v_1$ is an edge that connects $u_0v_1$ with $x_0y_1$ which are edges of the induced matching in $G \times e$, a contradiction.

For the second step, assume that $S$ is not an induced matching of $G \times e$.
It contains a pair of edges witnessing it is not induced -- $u_0v_1$ and $x_0 y_1$ such that (without loss of generality) $v_1$ is connected to $x_0$.
Then, $f(u_0v_1) = (u,v)$ and $f(x_0 y_1) = (x, y)$ such that $f(v_1 x_0)=(x,v)$, hence, edge $xv \in E(G)$ and $f(S)$ is not a split orientation. Hence, the lemma is proved. \qed

\subsection*{Details of the proof of \Cref{lem:max_leaf_so}}

\subsubsection*{Proof of \Cref{clm:core-has-split-orientation}}
Let $T$ denote a core. Let $A_\mathrm{min} = \{(u,v) \mid u \in V(T) \land \mathrm{leaf}(u)\}$.
We prove that $A_\mathrm{min}$ satisfies the split orientation properties.
Towards a contradiction, assume that arcs $(u,v)$ and $(w,q)$, such that $u$ and $w$ are leaves in $T$, witness that $A_\mathrm{min}$ breaks a split orientation property.
Then either $uq \in E(T)$ or $wv \in E(T)$ but both $u$ and $w$ are leaves, in either case we have $v=q$.
As both leaves $u,w$ are adjacent to $v=q$ this vertex would be adjacent to two leaves which contradicts $T$ being a core.
The set $A_\mathrm{min}$ satisfies the split orientation properties and contains all \emph{leaf arcs} so $T$ has a leafed split orientation $A_\mathrm{min}$.
Note that by definition any leafed split orientation of $T$ contains $A_\mathrm{min}$.
By the above argument we see that a tree has a leafed split orientation if and only if it is a core.
\qed

\subsubsection*{Proof of \Cref{clm:leaves-are-at-distance-4}}

As the leaves in $L_1$ come from vertices $u_0 \in T_0$ we know that their distance from $r$ is even. Hence, for any two leaves of $L_1$, their distance from each other is also even.
It clearly cannot be 0, and in case it is 2 we would have 2 leaves adjacent to the same parent which contradicts $T$ being a core. Thus, the claim holds for $L_1$. Analogously, we can argue the same for two leaves $L_2$ with the modification that distances from $r$ is odd for both, so distance between each other is even. Hence, the claim is proved. \qed
        
\subsubsection*{Proof of \Cref{clm:maximum-leafed-orientation}}

Due to the bijection $f$ (by \Cref{lem:bijection}) it suffices to show that $g_1(R_1) \cup g_2(R_2)$ is a maximum leafed induced matching.
This is exactly what \Cref{lem:leafedmimisniceandfast} computes.
As $g_1$ and $g_2$ are straight-forward bijections we end up with the fact that $g_1(R_1)$ (resp. $g_2(R_2)$) is a maximum induced matching of $T_1$ (resp. $T_2$) that contains all of the leaves in $g_1(L_1)$ (resp. $g_2(L_2)$).

We first show that $|T| \leq 5|A|-1 $.
As $A=f(g_1(R_1) \cup g_2(R_2))$ is constructed through bijections we have $|A|=|R_1|+|R_2|$.
\Cref{lem:leafedmimisniceandfast} implies that $|R_1|+|R_2| \ge |M|$, where $M$ is a maximum induced matching in $T$.
By \Cref{lem:mim_size_relation} we know that $|T| \le 5|M|-1$.
From these three results we can conclude that $|T| \le  5|M|-1 \le 5(|R_1|+|R_2|)-1 = 5|A|-1$.

For the other inequality, we note that each arc in $A$ has a head and a tail and by the definition of split orientation no two arcs can have heads (or tails) on the same vertex. Thus, it follows that $|T| \ge |A|$.
However, if $|A| \ge 3$ then $|T|\ge 3$ and there is an arc with one non-leaf endpoint.
This endpoint, say head (tail has a symmetric argument), cannot be adjacent to a tail vertex different from its own.
So the total number of tail vertices is strictly smaller than $|T|$, and $|A|+1 \leq |T|$, follows consequently. 
\qed

\subsection*{Proof of \Cref{thm:wone_complex_components}}

We establish $W[1]$-hardness via a reduction from \textsc{Multicolored Independent Set} (MIS).
  
We aim to produce a graph $G \in \mathcal C_k$ that comprises of 0-edges (which define $G_0$) and 1-edges (which define $G_1$) and models the MIS instance on graph $H$.
Let $\mathcal I \in \mathcal C_k$ be the instance that contains at least $k$ trees $T_1,\dots,T_k$ in its 1-edges such that for each $i \in [k]$ the tree $T_i$ has a core of size at least $5 n_0-1$
As $\mathcal C_k$ is closed under editing 0-edges let us begin by having $G_1$ equal to the 1-edges of $\mathcal I$, then we describe how to add 0-edges to $G_0$ to arrive at $G$ in the end; for which $G \in \mathcal C_k$ because it has identical 1-edges to $\mathcal I$.

Let $n_0 = \max_i |V_i|$ and apply \Cref{lem:core_procedure} to retrieve core $C_i$ for each $T_i$.
Due to \Cref{lem:max_leaf_so} the maximum leafed split orientation $A_i$ has size at least $n_0$.
For simplicity, assume for a moment that $|V_i|=|A_i|=n_0$ for each $i \in [k]$.
To each core $C_i$ we add edges as per \Cref{lem:whard_options} to fix roots $r_{i,1},\dots,r_{i,n_0}$ and their private neighborhoods $s_{i,1},\dots,s_{i,n_0}$.
Rooting $T_i$ at $r_{i,j}$ is meant to represent selecting $j$-th vertex of color $V_i$ in the MIS instance.
We prevent concurrent selection of roots that represent neighboring vertices of MIS by connecting their private neighborhoods as follows.

For every edge $uv \in V(H)$ in the MIS instance, let $i,i',j,j'$ be integers such that $u$ is the $j$-th vertex of $V_i$, and $v$ is the $j'$-th vertex of $V_{i'}$.
To $G_0$, we add a 0-edge from $s_{i,j}$ to $s_{i',j'}$.

To show equivalence in one direction, assume that we have a yes-instance of MIS with a solution $S$.
For each $i \in [k]$, let $v_{i,j}$ be the $j$-th vertex of $V_i$ such that $v_{i,j} \in S$, and let us root $C_i$ in $r_{i,j}$.
Towards a contradiction, assume that this is not a valid rooting.
Then there exists a pair $r_{i,j}$, $r_{i',j'}$ such that their neighborhoods do not form an independent set.
As both $r_{i,j}$ and $r_{i',j'}$ have an independent neighborhoods of their own, the violating 0-edge must go from one neighborhood to the other.
The trees $T_i$ and $T_{i'}$ are connected only with 0-edges between private neighborhood vertices and both $r_{i,j}$ and $r_{i',j'}$ are adjacent only to $s_{i,j}$ and $s_{i',j'}$ from this set, so there is a 0-edge $s_{i,j}s_{i',j'}$.
This edge was created due to an edge $v_{i,j}v_{i',j'} \in E(H)$ which witnesses that $S$ is not a solution, a contradiction.

In the other direction, assume that we have a rooting of the $T_i$s as a solution to the \EFXO instance.
By the construction the only roots that are feasible (have no 0-edges between their neighbors) are those returned from \Cref{lem:whard_options}.
For any $i,i' \in [k]$ for the two selected roots $r_{i,j},r_{i',j'}$ there is no 0-edge $s_{i,j}s_{i',j'}$ so pulling these vertices back to MIS by selecting $v_{i,j}$ which is the $j$-th vertex of $V_i$ and $v_{i',j'}$ which is the $j'$-th vertex of $V_{i'}$, there is no edge $v_{i,j}v_{i',j'}$ in $H$.

For the general case where we have $|V_i| \le n_0 \le |A_i|$ for each $i \in [k]$ we want the extra roots to mimic the last root.
We take the last $|A_i|-|V_i|$ roots of $A_i$ and connect their private neighborhoods to wherever the private neighborhood of is connected $r_{i,|V_i|}$.
More specifically, we process these extra roots one by one for each $i \in [k]$, this for example results in an edge $r_{i,|V_i|}r_{i',|V_{j'}|}$ creating a complete bipartite graph between the respective extra roots.
As the private neighborhoods of these extra roots are exactly the same as the private neighborhood of the root they are mimicking any solution that roots the tree in them can be altered to root in $r_{i,|V_i|}$ so the proof still holds in the general case. \qed

\end{document}